\documentclass[pre,superscriptaddress,nofootinbib,notitlepage,twocolumn,tightenlines,amsmath]{revtex4-1}
\usepackage{graphicx,isomath,bm}
\usepackage{subfiles}
\usepackage{amsmath,amsfonts,amsthm,bbm,mathtools}
\usepackage[colorlinks=true,linkcolor=blue,citecolor=blue,urlcolor=blue,pdfborderstyle={/S/U/W 1}]{hyperref}

\begin{document}

\title{Signature of Glassy Dynamics in Dynamic Mode Decompositions}
\author{Zachary G. Nicolaou}
\affiliation{Department of Applied Mathematics, University of Washington, Seattle, Washington 98195, USA}
\affiliation{Department of Physics \& Astronomy, University of Tennessee, Knoxville, Tennessee 37996, USA}
\affiliation{Department of Mathematics, University of Tennessee, Knoxville, Tennessee 37996, USA}
\author{Hangjun Cho}
\affiliation{Department of Mechanical Engineering, University of Washington, Seattle, Washington 98195, USA}
\affiliation{AI Institute in Dynamic Systems, University of Washington, Seattle, WA 98195, USA}
\affiliation{Research Institute of Mathematics, Seoul National University, Seoul 08826, Republic of Korea}
\author{Yuanzhao Zhang}
\affiliation{Santa Fe Institute, Santa Fe, New Mexico 87501, USA}
\author{J. Nathan Kutz}
\affiliation{Department of Applied Mathematics, University of Washington, Seattle, Washington 98195, USA}
\affiliation{Department of Electrical and Computer Engineering, University of Washington, Seattle, Washington 98195, USA}
\affiliation{AI Institute in Dynamic Systems, University of Washington, Seattle, WA 98195, USA}
\author{Steven L. Brunton}
\affiliation{Department of Mechanical Engineering, University of Washington, Seattle, Washington 98195, USA}
\affiliation{AI Institute in Dynamic Systems, University of Washington, Seattle, WA 98195, USA}

\begin{abstract}
Glasses are traditionally characterized by their rugged landscape of disordered low-energy states and their slow relaxation towards thermodynamic equilibrium. Far from equilibrium, dynamical forms of glassy behavior with anomalous algebraic relaxation have also been noted, for example, in networks of coupled oscillators. Due to their disordered and high-dimensional nature, such systems have been difficult to study theoretically, but data-driven methods are emerging as a promising alternative that may aid in their analysis. Here, we characterize glassy dynamics using the dynamic mode decomposition, a data-driven spectral computation that approximates the Koopman spectrum. We show that the gap between oscillatory and decaying modes in the Koopman spectrum vanishes in systems exhibiting algebraic relaxation, and thus, we propose a model-agnostic signature for robustly detecting and analyzing glassy dynamics. We demonstrate the utility of our approach through both a minimal example of a one-dimensional ODE and a high-dimensional example of coupled oscillators.   
\end{abstract}
\date{\today}

\maketitle

Quantitative characterization of glassy behavior is a longstanding scientific challenge. In material physics, glasses occur in systems with a plethora of nearly degenerate low-energy configurations, leading to slow relaxation towards equilibrium, spatial disorder, and hysteresis. Frustration, wherein disharmonious interactions prevent system components from simultaneously realizing their lowest-energy configurations, often leads to glassy behavior. Progress has been made in characterizing the thermodynamics of some glassy systems such as spin glasses \cite{1987_Mezard,2003_Mezard}, but significant challenges remain, especially in systems far from equilibrium.  

Dynamical analogs of glassy behavior have long attracted interest in oscillatory systems, for example, in the frozen vortex glass of the complex Ginzburg-Landau equation \cite{2003_Brito,2017_Nicolaou}. In networks of coupled oscillators, Daido first detailed a system \cite{daido} in which the dynamics relax towards an incoherent attractor at a slow, algebraic rate, in stark contrast to the exponential relaxation typically observed for attractors in dynamical systems.  More recently, \cite{strogatz}, Ottino and Strogatz revisited Daido's model and recognized challenging, high-dimensional features in the dynamics in the glassy regime and a related and interesting volcano transition \cite{Pruser24, Pazo23, Lee24}.

\begin{figure}[b]
\includegraphics[width=\columnwidth]{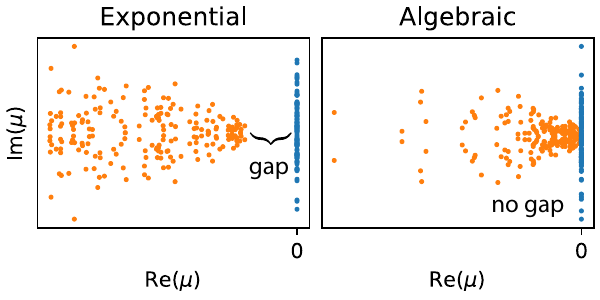}
\caption{Schematic of Koopman oscillatory (blue) and decaying (orange) eigenvalues for systems exhibiting exponential (left) and algebraic (right) decay.  The signature absence of a gap in the spectrum allows one to detect glassy dynamics. \label{fig1}}
\end{figure}
In this Letter, we propose a general strategy to detect glassy dynamics and non-exponential relaxation using the dynamic mode decomposition (DMD) \cite{Rowley09,Schmid10,Kutz16,brunton,colbrook23,Otto21}, a data-driven framework for characterizing nonlinear dynamical systems by approximating the associated, inifinte-dimensional Koopman and Perron-Frobenious linear operators \cite{Koopman31, 1980_Landau, Singh93, Mezic05, Budisic12, Bollt18, Mauroy16, Bruce19, Mezic20, Mauroy20,colbrook24-limit, Morrison24} (see Supplemental Material Sec.~S1A). DMD has shown remarkable success in a variety of applied and theoretical studies \cite{Schmid10,Rowley09,Schmid22,Rupe24,brunton} and is also being explored in the study of coupled oscillators \cite{Curtis21,Mauroy13,Hu20,Hashimoto22,Susuki11,Wilson22,Thibeault25}.  
Such glassy dynamics can be hidden in high-dimensional data and might be easy to miss from conventional analyses, especially if one lacks {\em a priori} knowledge of a good order parameter with which to quantify the scaling. In these cases, our DMD analysis offers a principled and system-agnostic way to extract algebraic relaxation dynamics from high-dimensional data. 

Specifically, we note a key signature of glassy dynamics in the Koopman spectrum, as illustrated schematically in Fig.~\ref{fig1}. 
For bounded systems, the Koopman spectrum can be divided into an oscillatory component, whose (continuous-time) eigenvalues are purely imaginary, and a decaying component, whose eigenvalues lie in the left half-plane (see Supplemental Material Sec.~S1B). If a gap exists between the oscillatory and the decaying modes, the relaxation towards the long-term behavior will eventually be dominated by the leading decaying mode with the smallest decay rate $-\mathrm{Re}(\mu)$. However, if this gap vanishes and there is an accumulation of decaying modes approaching the oscillatory component, other asymptotic behaviors are possible. 

In the case of a vanishing gap, appropriate scaling between the mode amplitude and the vanishing decay rate can enable the cumulative response of exponentially decaying modes to give rise to algebraic relaxation. 
For example, Ogielski and Stein propose to decompose a power law into an infinite number of exponentially decaying modes \cite{Ogielski85}. Mathematically, they show that the infinite sum $P(t) = \sum_{m=0}^\infty w_m e^{-\lambda_m t}$ asymptotes to $P(t) \sim t^{-\alpha}-O(e^{-t}/t)$ as $t\to\infty$ if we set $w_m = 2^{-m}$ and $\lambda_m = \frac{c(1+c)}{1-c}c^m$, where $\alpha=\frac{\log2}{\log c^{-1}}>0$ and $0<c<1$ are constants. 
Thus, observing a similar scaling in the DMD spectrum, which approximates the Koopman spectrum, and DMD mode amplitudes can serve as a principled way to determine if glassy dynamics are present in the data.  

For clarity of presentation, we demonstrate our results first on an idealized one-dimensional system before focusing in detail on the high-dimensional Daido glassy oscillator model. We demonstrate how the glassy signature enables the development of a data-driven order parameter characterizing the transition to glassy behavior based on the distribution of DMD eigenvalues. Finally, we discuss potential applications and extensions for other systems. All the numerical results in this Letter can be reproduced from our GitHub repository \cite{github}.

\textbf{Dynamic Mode Decomposition.}
We first briefly outline the DMD methods employed in our analyses. DMD gives a matrix approximation $\mathbb{K}$  of the Koopman operator 
corresponding to a dynamical system $\dot{x}=f(x)$, which is the infinite-dimensional linear composition operator governing the temporal evolution of measurement functions (see Supplemental Material, Sec.~S1A).  We are interested in the continuous case here, but data is typically available for a temporal discretization. The state $x$ here may be a scalar, a vector, or a spatiotemporal field. 
We generate multiple solution trajectories $x^j(t)$, indexed by the superscript $j$, in an effort to sample the state space well. In the extended DMD approach \cite{Williams15} that we employ, we introduce a dictionary of measurements $\mathcal{D}=\{g_1,\ldots,g_d\}$ in which we seek to observe this linear dynamics, where $g_i:\mathcal{M}\to\mathbb{C}$ are observable functions of the state variables and $\mathcal{M}$ is the state space. We consider the evolution of the vector composed of the dictionary evaluated on the data $\vec{\mathcal{D}}(x^j(t))$, concatenated over all trajectories, where the trajectories are sampled in time steps of size $dt$. The matrix $\mathbb{K}$ is the optimal linear solution for the time-varying dynamics of the dictionary terms; $\vec{\mathcal{D}}(x^j(t+dt)) \approx \mathbb{K}  \vec{\mathcal{D}}(x^j(t))$. The quality of these approximations and their convergence depend on how close the span of the dictionary terms is to a Koopman invariant subspace \cite{Otto19, Haseli22} and on how well the data represents the full dynamics of the state space \cite{Korda18}.

The DMD analysis produces an eigendecomposition of the matrix $\mathbb{K}= \Phi \Lambda \Phi^{-1}$, where $\Phi= [\phi_{\lambda_1},\ldots,\phi_{\lambda_d}]$ is the invertible matrix of eigenvectors (whose matrix elements we will denote by $\Phi_i^k$) and $\Lambda=\mbox{diag}(\lambda_1,\ldots,\lambda_d)$ is the diagonal matrix of eigenvalues, which can be used to approximate Koopman functions as $\phi_i(x) = \sum_k \Phi_i^k g_k(x)$. We can  then evaluate the time-evolution of the mode amplitudes $b_i^j(t) \equiv \phi_i(x^j(t))$, which are expected to evolve exponentially according to $b_i^j(t) \approx b_i^j(0) e^{\mu_i t}$, where $\mu_i = \log(\lambda_i)/dt$ is the continuous-time eigenvalue. If the approximations are good, the mode amplitudes will indeed evolve exponentially, and the evolution of the dictionary terms can then be reconstructed as $\hat{g}_k(x^j(t)) \approx \sum_i \tilde{\Phi}_k^i b_i^j(0)e^{\mu_i t}$ where $\tilde{\Phi}_k^i$ are the matrix elements of the inverse $\Phi^{-1}$ (which is also the matrix whose rows are the left eigenvectors of $\mathbb{K}$). Truncated reconstructions can be used to significantly compress the data required to represent the dynamics, provided the fit is good (i.e., the mode amplitude evolution is, in fact, exponential).

The basic version of extended DMD described above suffers from issues such as spectral noise in practical applications. To mitigate these issues, we employ two extensions of the method. First, we use the exact DMD approach \cite{Tu14}, which reduces the numerical dimension of the problem (and thus reduces the impacts of finite precision) by first projecting the analysis onto a minimal set of principal modes given by the singular value decomposition.  We also employ the residual DMD (resDMD) approach \cite{colbrook24,colbrook24-res}, which evaluates the quality of each mode through a residual that vanishes for true Koopman modes. The resDMD analysis also enables us to evaluate pseudospectra approximating the norm $\varepsilon$ of the resolvant of the Koopman operator, which diverges exactly on the Koopman spectrum. The pseudospectrum can help to identify neighborhoods close to Koopman eigenvalues even when the actual Koopman mode does not lie in the span of the dictionary. 

\textbf{Algebraic decay in a minimal example.}
We now consider the DMD analysis in a one-dimensional system given by
\begin{align}
\dot{\theta}&=-\sin(\theta)^\zeta. \label{eq2}
\end{align}
This example clearly illustrates the difference between exponential relaxation ($\zeta=1$) and algebraic relaxation ($\zeta=3$). The only attractor for this system is $\theta=0$.
It is easy to show analytically that the state decays asymptotically according to $\theta(t)\sim e^{-t}$ for $\zeta=1$ and $\theta(t)\sim t^{-1/2}$ for $\zeta=3$ and that a gap is present in the Koopman spectrum only in the $\zeta=1$ case (see Supplemental Material, Sec.~S2A). 
\begin{figure}[t]
\includegraphics[width=\columnwidth]{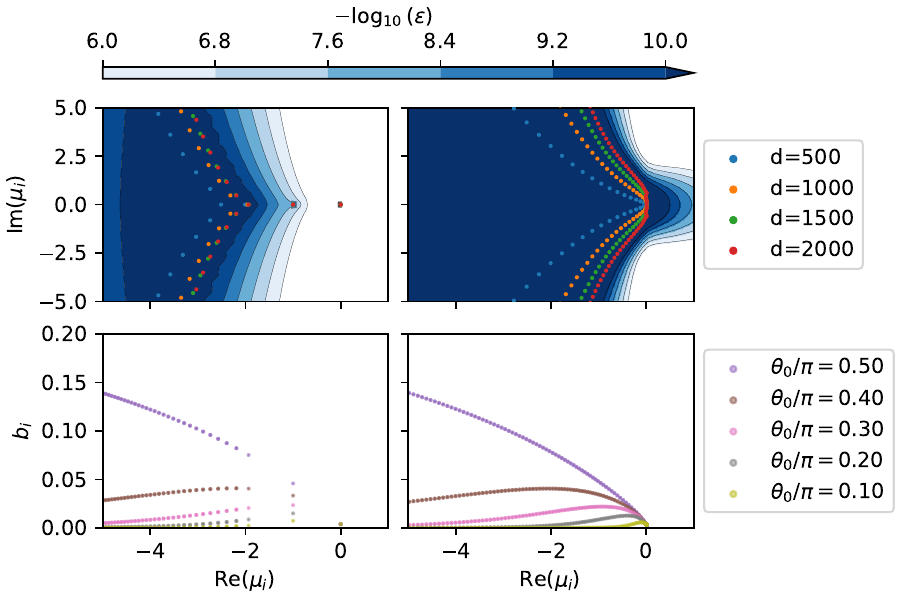}
\caption{DMD spectrum and pseudospectrum (top panels) and mode amplitudes vs decay rate (bottom panels) for Eq.~\eqref{eq2} with $\zeta=1$ (left panels) and $\zeta=3$ (right panels).  Only in the $\zeta=3$ spectrum do we observe the glassy DMD signature (i.e., an accumulation of decaying modes approaching the imaginary axis). \label{fig2}}
\end{figure}

We employ an extended resDMD analysis with a dictionary composed of Fourier-based measurements $g_k(\theta)\equiv e^{\imath k\theta }$ for integers $k$ with $-d \leq k \leq d$ (where $\imath$ is the imaginary unit), which can approximate arbitrary square-integrable 2$\pi$-periodic functions well for large $d$.  We perform the fit with combined data from $5$ trajectories $\theta_j(t)$ with initial conditions $\theta_j(0)/\pi=(0.5, 0.4, 0.3, 0.2, 0.1)$.  We confirm that the input trajectories for both $\zeta =1$ and $\zeta=3$ can be successfully reconstructed from the initial mode amplitudes and the approximate Koopman modes (see Supplemental Material Sec.~S2B), suggesting that our dictionary is sufficiently expressive.

We turn next to the qualitative features of the DMD spectrum in the two cases, as shown in Fig.~\ref{fig2}. Since the asymptotic state for Eq.~\eqref{eq2} is a fixed point, there is no oscillatory component, and the DMD spectrum consists only of decaying modes with $\mathrm{Re}(\mu)<0$. As noted above, the algebraically decaying trajectories are well-approximated by a sum of exponentially decaying Koopman modes, even though any finite sum of exponentially decaying terms will eventually decay exponentially. This apparent contradiction can be resolved by noting that the number of terms required in the exact Koopman representation may diverge provided the amplitudes decay sufficiently fast and the sum converges, as depicted in Fig.~\ref{fig1}. 

Indeed, the key difference between the spectra in Fig.~\ref{fig2} is the presence of a gap between the imaginary axis and the spectrum in the exponential case. When the gap vanishes, we also observe an associated widening and bulging of the pseudospectrum around the imaginary axis. Incidentally, we note that in the study of numerical instabilities, pseudospectral bulges like the one seen in Fig.~\ref{fig2} have also been related to algebraic growth \cite{Trefethen05}. Finally, in the algebraically decaying case, we observe a scaling between the initial mode amplitudes $b_i^j(0)$ and the decay rate $-\mathrm{Re}(\mu)$ that enables the sum of exponential terms to approximate an algebraic decay (see Supplemental Material Sec.~S2B), in contrast to the exponential case, where the leading mode quickly dominates.

\begin{figure}[b]
\includegraphics[width=\columnwidth]{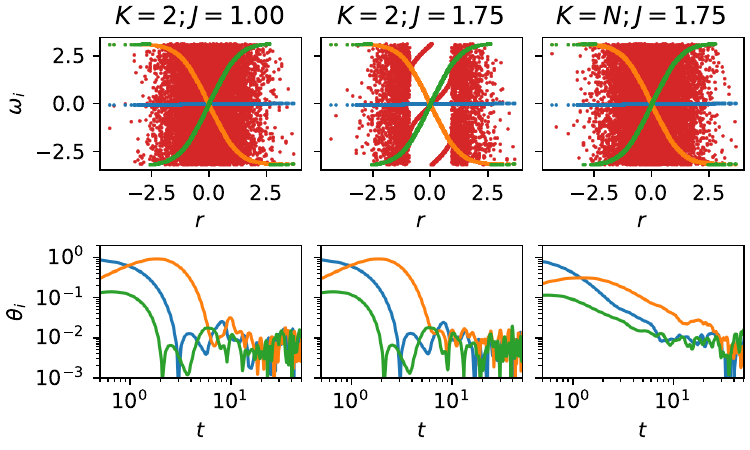}
\caption{Dynamics of three systems of $N=10000$ oscillators in Eq.~\eqref{eq1} with various coupling constants and rank, as indicated above each column. Top panels show oscillator phases vs natural frequency for the initial conditions for each training trajectory (blue, orange, and green dots) and the final state (red dots), and bottom panels show Kuramoto order parameter vs time for each training trajectory. The low rank systems ($K=2$) exhibit exponential relaxation, while the high rank system ($K=N$) exhibits algebraic relaxation. }
\label{fig3}
\end{figure}

\begin{figure*}
\includegraphics[width=2\columnwidth]{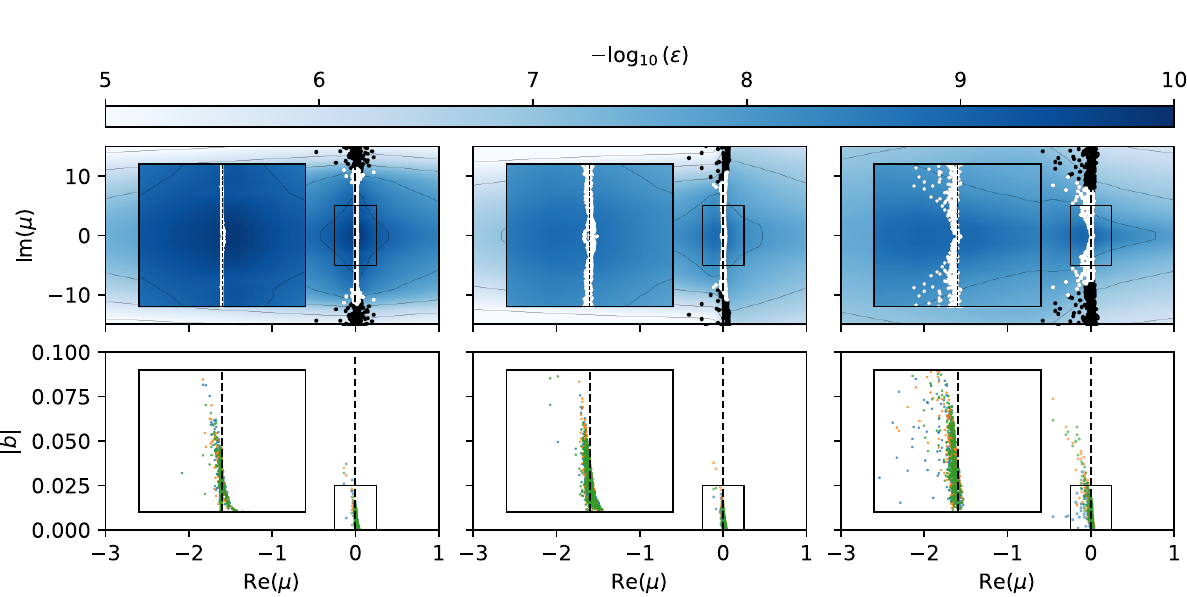}
\caption{DMD spectra for the three systems in Fig.~\ref{fig3} with $K=2$, $J=1.00$ (left panels), $K=2$, $J=1.75$ (middle panels), and $K=N$, $J=1.75$ (right panels). In the top panels, the white dots show non-spurious DMD eigenvalues that pass the resDMD criteria  $\varepsilon<5\times 10^{-8}$, black dots show spurious DMD eigenvalues, and blue shading shows the approximate $\varepsilon$-pseudospectra. Insets show a zoom of the boxed portions (continuous black lines) near the imaginary axis (dashed black lines). 
Bottom panels show the mode amplitude $|b_i^j|$ vs decay rate $\mathrm{Re}(\mu_i)$, with color corresponding to the trajectories in Fig.~\ref{fig3}. The spectrum is confined to the imaginary axis in the low rank ($K=2$) cases, but the glassy signature of accumulating decaying modes is present in the high rank ($K=N$) case.  
\label{fig4}}
\end{figure*}

\textbf{Algebraic decay in the oscillator glass.}
We next consider the generalized Daido model of an oscillator glass proposed in Ref.~\cite{strogatz}, 
\begin{align}
\label{eq1}
\dot{\theta}_n &= \omega_n +J\sum_{m=1}^N A_{nm}\sin\left(\theta_m-\theta_n\right), 
\end{align}
which describes the evolution of $N$ coupled phase oscillators with frequencies $\omega_n$ (sampled here from a standard normal distribution), $J$ represents the coupling constant, and $A_{nm} \equiv \sum_{k=1}^{K}(-1)^k u_k^{(n)}u_k^{(m)}/N$ are the components of a random adjacency matrix for a given even integer $K>0$. Here, $u_k^{(n)}$ are the components of a $K$-dimensional random vectors with each element independent and equal to $\pm 1$ with equal probability. In the large $N$ limit, $K$ is the rank of the adjacency matrix. Since the signs of the adjacency matrix elements are random, the couplings include both attractive and repulsive terms, leading to frustration.  In the low-rank regime, an interesting volcano transition has been characterized by a low-dimensional Ott-Antonsen reduction \cite{Ott08, Ott09}. However, Daido's glassy states appear only in the high-rank regime, where no dimension reduction is available.

We implement an efficient GPU integration scheme to generate data from $N=10000$ coupled oscillators at different combinations of coupling constants and adjacency matrix ranks. Our simulations enable us to investigate larger systems than previously studied, and we note a few new results about the large $N$ systems in Sec.~S3A of the Supplemental Material. Different relaxation behaviors are observed in various parameter regimes, as shown in Fig.~\ref{fig3}. We produce trajectories from three initially ordered states that relax towards disordered states. These initial conditions consist of identical initial phases $\theta_n(0)=0$ and initial phases with a single winding around the circle $\theta_n(0)=\pm (-\pi+2\pi n/N)$, where the oscillator index is ordered by the natural frequencies. The Kuramoto order parameter,
\begin{equation}
\label{kuramoto}
r = \left \lvert\frac{1}{N}\sum_{n=1}^N e^{\imath\theta_n}\right \rvert,
\end{equation}
has traditionally been used to quantify the relaxation. For low-rank systems with a low coupling strength (left column in Fig.~\ref{fig3}), initial conditions decay exponentially towards an incoherent state. For low-rank systems above the critical coupling threshold (middle column in Fig.~\ref{fig3}), interesting volcano clusters appear in the final state, but the relaxation remains exponential. Finally, for high-rank systems with a large coupling constant (right column in Fig.~\ref{fig3}), the slow relaxation toward incoherence follows a power law, exhibiting glassy dynamics. 

We can again express arbitrary periodic measurement functions with Fourier basis elements $g_{\nu_k}(\theta)=e^{\imath \sum \nu_k \theta_k}$ for each multi-index $\nu_k\in \mathbb{Z}^N$. Indeed, in the uncoupled ($J=0$) case, each Fourier basis element $g_{\nu_k}$ is an exact oscillatory Koopman eigenfunction with eigenvalue $\mu_{\nu_k}=\imath \sum_k \nu_k \omega_k$ (see Supplemental Material Sec.~S3B). We base our extended DMD dictionary on this basis, but the high dimensionality of the state space imposes significant limitations on our dictionary. We restrict our dictionary to include only some of the lowest-order multi-indices. To enable us to study the dependence with increasing dictionary size, we include those Fourier modes with $\nu_k = \pm m\delta_k^n$, where $\delta_k^n$ is the Kronecker delta, $n=1,2,\cdots,N$, and $m=1,2,\cdots M$, leading to $2NM$ dictionary terms. We also include a random selection $NM$ of pairwise terms of the form $\nu_k = \pm \delta_{i(k)}^n \pm \delta_{j(k)}^n$, where $i(k)$ and $j(k)$ are sampled from all pairs of distinct choices without replacement, leading to a total dictionary size of $D=3NM$ terms. 
Figure~\ref{fig4} shows our results for the largest dictionary considered with $D=200000$ terms for the three trajectories from Fig.~\ref{fig3}. Since the attractors are not steady states, we expect the presence of oscillatory modes in the Koopman spectra, which we indeed observe in the resDMD approximations close to the imaginary axis. 

Qualitative differences between the results for the exponentially relaxing cases and the glassy, algebraic relaxation are apparent in the results, but the interpretation is somewhat confounded by spectral noise in the DMD approximation of the Koopman spectrum. Nevertheless, we can clearly observe the convergence of the oscillatory component of the spectrum to the imaginary axis as the dictionary size increases (see Supplemental Material Sec.~S3C).  Only in the glassy case with algebraic decay (right column in Fig.~\ref{fig4}) do we see an additional accumulation of decaying modes near the imaginary axis that do not converge to the imaginary axis as the dictionary size increases.  Indeed, the DMD spectrum faithfully captures the exponential and algebraic behavior in the reconstruction of the Kuramoto order parameter,
 \begin{align}
\hat{r}^j(t) = \sum_{i,k} \left|\frac{1}{N} \tilde{\Phi}_k^ib_i^j(0)e^{\mu_i t}\right|, \label{kuramoto}
\end{align}
and thus the DMD reconstruction can effectively represent the glassy dynamcis observed in the oscillator system. 
To avoid overfitting and to mitigate spectral noise in the reconstruction, we include only the minimum necessary number of small-residue modes to reconstruct the Kuramoto order parameter to a relative accuracy of $10^{-3}$, with a residue threshold of $\varepsilon<5\times 10^{-8}$ (white dots in Fig.~\ref{fig4}). 

It is interesting to note that in the nonglassy regimes (left two columns in Fig.~\ref{fig4}), the DMD spectrum does not appear to include any decaying modes and consists instead of a large number of purely oscillatory modes. The incoherent summation of these oscillator modes can nevertheless give rise to an exponential relaxation towards the incoherent state. This can occur via the generalized Landau damping phenomenon noted previously in the nonlinear Fokker-Planck description of the Kuramoto system \cite{strogatz_1992}, where an analytical continuation argument reveals ``fake'' (decaying) eigenvalues responsible for the decay. (We note an explicit relationship between our DMD spectrum and the nonlinear Fokker-Planck equation in the Supplemental Material Sec.~S3D.) In the glassy regime, on the other hand, we do observe the signature of a large number of decaying DMD modes accumulating near the imaginary axis. Crucially, in both the glassy and nonglassy regimes, we emphasize that the distribution of decaying modes in our order parameter reconstruction reproduces the observed exponential and algebraic decay rates very accurately (see Supplemental Material Sec.~S3C).

While the asymptotic behavior in the Kuramoto order parameter has traditionally been used to assess the potential for glassy dynamics, several issues have become apparent that have complicated analysis, stymied development in the literature, and led to contentious and unproductive discussions. Since only finite-size systems can be computationally simulated, the thermodynamic asymptotic relaxation is only approximately observed over a finite time window before the Kuramoto order parameter approaches the finite-size noise floor. Thus, the exact location in parameter space deliminating the glassy dynamical behavior has remained obscured.  Rather than using the DMD reconstruction of the Kuramoto order parameter in Eq.~\eqref{kuramoto}, we introduce a new, data-driven order parameter, which better quantifies the onset of glassy dynamics in the oscillator glass.  Our data-driven order parameter $\eta$ is simply given by the average negative real part of the DMD eigenvalues that pass the residual DMD criteria $\varepsilon<5\times10^{-8}$ (i.e., the white dots in Fig.~\ref{fig4}),
\begin{equation}
\eta \equiv -\langle \mathrm{Re} (\mu_i) \rangle \rvert_{\varepsilon_i<5\times10^{-8}}.
\end{equation}

While we would prefer to evaluate $\eta$ using the exact Koopman spectrum, in practice, we can only calculate the DMD eigenvalues, which, as noted above, are close to but not exactly equal to the Koopman eigenvalues. Given this source of spectral noise, some DMD eigenvalues actually appear to the right of the imaginary axis, implying a ``forbidden'' exponentially-growing DMD mode in a bounded system.  Such eigenvalues can even pass the residual DMD test as long as the inverse of their positive real part is small compared to the time length of the input trajectories. We can thus take the average real part of those positive DMD eigenvalues as a crude estimate for the spectral uncertainty for the individual eigenvalues, which we divide by the square root of the number of eigenvalues included in the mean and use to construct random error bounds for the order parameter. 

\begin{figure}
\includegraphics[width=\columnwidth]{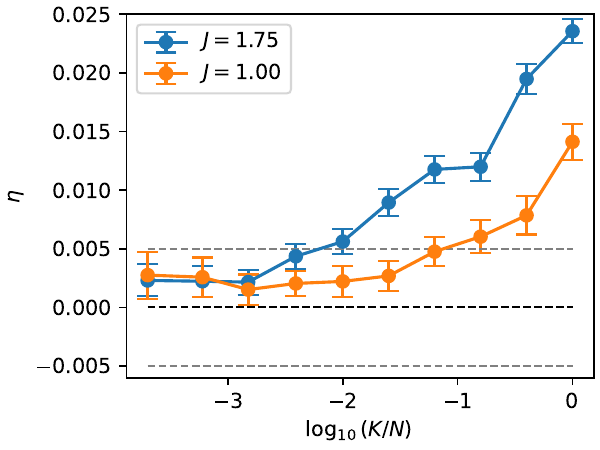}
\caption{Data-driven order parameter $\eta$ vs logarithmic rank $\log_{10}(K/N)$, signaling the onset of glassy dynamics. Error bars show the estimated order parameter random error bounds derived from average of the DMD eigenvalues with positive real parts, and the gray dashed lines show the more conservative finite-size systematic error bounds. \label{fig_order}}
\end{figure}

Figure \ref{fig_order} shows the order parameter $\eta$ vs. the rank $K$ of the network for two systems with different coupling constants. Note that the order parameter $\eta$ is typically slightly larger than zero (black dashed line in Fig.~\ref{fig_order}) even for small rank systems.  We attribute this nonzero value to systematic errors in the DMD spectrum due to finite-size effects, including the fact that the span of the finite dictionary terms does not exactly contain any Koopman invariant subspaces and the trajectories only span a finite time. These systematic errors appear to be conservatively bounded by a window of size equal to the smallest detectable decay rate $1/T$ determined by the time span of the input trajectories (grey dashed lines in Fig.~\ref{fig_order}). It is clear that the order parameter is nonzero for systems with sufficiently large rank, and, furthermore, that this onset of the glassy dynamics occurs at a larger rank when then coupling constant is smaller. Thus, the transition to glassy decay can be quantitatively delimited using our data-driven order parameter.

\textbf{Discussion.}
Leveraging a key feature in the Koopman spectrum (the presence or absence of a spectral gap), the dynamic mode decomposition and pseudospectral estimates can be utilized to define a data-driven order parameter for distinguishing exponential and algebraic dynamical relaxation, as relevant in the study of glassy systems far from equilibrium. Our results support previous findings \cite{daido,strogatz} of glassy dynamics in Daido's frustrated oscillator model, which were only established through significant domain expertise and specialized knowledge of coupled oscillators. While we focused on Daido's oscillator glass here for clarity, the approach of deriving order parameters from DMD spectrum or other data-driven methods can be relevant to other systems as well, and we expect it may be fruitfully applied in the study of, e.g., other disordered systems of oscillators \cite{Nakagawa93} and neural networks \cite{Kadmon15,Mastrogiuseppe18}. This work contributes to the growing body of literature utilizing data-driven approaches for qualitative system characterization \cite{Bruton17, Bollt18, Gottwald20, Hashimoto22, Rupe24,colbrook24-res}. 

One central challenge in our approach is the uncertainty inherent in the DMD estimates of the Koopman spectrum, which obscures the distinction between the oscillatory and decaying parts of the spectrum. Future studies would benefit from better theoretical bounds for quantifying the uncertainty of the DMD estimates. Bootstrapping or bagging estimates of the uncertainty may aid in this regard. Finally, we note that other data-driven techniques, such as applying DMD eigenvalue constraints \cite{Kutz18}, physics-informed constraints \cite{Yin24}, and sparse system identification \cite{Dawson22}, could prove useful in future studies. 

\textbf{Acknowledgements.}
The authors gratefully acknowledge Steven H. Strogatz and Matthew J. Colbrook for insightful discussions related to this work. Funding was provided by the Washington Research Foundation Postdoctoral Fellowship [ZGN], the National Research Foundation of Korea (NRF) grant funded by the Korea government (MSIT) (RS-2023-00253171) [HC], the Santa Fe Institute Omidyar Fellowship and National Science Foundation (NSF DMS 2436231) [YZ], and the National Science Foundation AI Institute in Dynamic Systems (grant number 2112085) [HC, JNK, SLB].

\clearpage

\fontsize{11}{2}
\setcounter{figure}{0}    
\setcounter{equation}{0}    
\fontfamily{cmr}\selectfont
\renewcommand{\baselinestretch}{1.5} 
\onecolumngrid
\renewcommand{\thefigure}{{S\arabic{figure}}}
\renewcommand{\thesection}{{S\arabic{section}}}
\renewcommand{\theequation}{{S\arabic{equation}}}
\newtheorem{theorem}{Theorem}
\newtheorem{lemma}{Lemma}

\begin{center}
\bf \normalsize \MakeUppercase {Supplemental Material to ``Signature of Glassy Dynamics in Dynamic Mode Decompositions''}
\date{\today}
\end{center}

\tableofcontents

\section{DMD and Koopman analysis}
In this section, we first provide a general description of the Koopman operator and next note some specific results for the class of Kuramoto systems considered in the main text.

\subsection{Brief introduction to Koopman operator and DMD}

\subsubsection{Koopman operator}
Let $(\mathcal{M},\mathcal{A},m)$ be the state space with the Lebesgue measure $m$ where $\mathcal{M} \subset \mathbb{R}^N$ (or the torus $\mathbb{T}^N$). We consider the discrete-time system 
\begin{align}\label{A-10}
\mathbf{x}(n+1) = T(\mathbf{x}(n)),\quad \mathbf{x}(0) = \mathbf{x}_0,
\end{align}
where $T:\mathcal{M}\to\mathcal{M}$ is a transformation. 
Let $\mathcal{F}$ be the set of (real or complex-valued) measurable functions on $\mathcal{M}$ and suppose that the transformation $T$ is non-singular. A measurable transformation $T$ is said to be non-singular if $m(T^{-1}(S)) = 0$ whenever $m(S)=0$, $S\in\mathcal{A}$.
The a linear transformation $U_T:{\rm L}^2(\mathcal{M})\to {\rm L}^2(\mathcal{M})$, called a composition operator or a Koopman operator \cite{Koopman31}, defined as the following is well-defined \cite{Singh93}: 
\begin{align*}
U_T g \coloneqq g\circ T,\quad g\in\mathcal{F}.
\end{align*}
Now, we consider an autonomous dynamical system
\begin{align}\label{SM-A-1}
	\dot{\mathbf{y}}(t) = F(\mathbf{y}(t)),\quad t\geq 0,\quad \mathbf{y}(0)=\mathbf{y}_0,
\end{align}
with a vector field $F:\mathcal{M}\to\mathcal{M}$ assuming that the vector field $F$ is sufficiently regular. We denote the flow by $\Phi^{t}:\mathcal{M}\to\mathcal{M}$ defined as $\Phi^{t}(\mathbf{y}_0) = \mathbf{y}(t)$, i.e., $\Phi^t(\mathbf{x}) = \mathbf{y}(t)$ for $\mathbf{x}\in\mathcal{M}$ where $\mathbf{y}(t)$ is the value of the flow $t\mapsto \mathbf{y}(t)$ starting at $\mathbf{y}(t_0)=\mathbf{x}$. 
Fix a constant ${dt}>0$.
When the map $\Phi^{dt}$ is non-singular, the Koopman operator $\mathcal{K}^{dt}:{\rm L}^p(\mathcal{M}) \to {\rm L}^p(\mathcal{M})$ is defined as $\mathcal{K}^{dt} g \coloneqq g\circ \Phi^{dt}$. We call such functions $g$ observables in the DMD context. We note that the Koopman operator is an infinite-dimensional linear operator, and we will study its spectrum.
With a fixed time $dt>0$, let $\varphi$ be an eigenfunction of the Koopman operator with an eigenvalue $\lambda$. Then,
\begin{align*}
\varphi(\mathbf{x}_k) = (\mathcal{K}^{dt})^{(k)} \varphi(\mathbf{x}_0) = \lambda^k \varphi(\mathbf{x}_0),\quad k=1,2,\ldots,
\end{align*}
where $f^{(k)} = f\circ \cdots \circ f$. In the continuous time case, in literature, an eigenfunction $\varphi$ with an eigenvalue $\mu$ is defined as
\begin{align*}
\mathcal{K}^t \varphi(\mathbf{x}) = e^{\mu t}\varphi(\mathbf{x}),\quad t\geq0,\quad \mathbf{x}\in\mathcal{M}.
\end{align*}
We will call $\mu$ the continuous eigenvalue in the main context.

We can consider a continuous-time formulation by directly evaluating the discrete-time Koopman operator for infinitesimal $t=dt$, yielding the eigenpair $(\varphi,e^{\mu {dt}})$ of $\mathcal{K}^{dt}$. The continuous-time formulation is useful, as it allows us to pose the eigenfunction problem for continuously differentiable, non-constant eigenfunctions via the differential equation:
\[
\frac{d}{dt}\mathcal{K}^t\varphi(\mathbf{x}(0)) = \frac{d}{dt}\varphi(\mathbf{x}(t)) = \nabla \varphi(\mathbf{x}(t)) \cdot F(\mathbf{x}(t)).
\]
In particular, $\frac{d}{dt}\varphi(\mathbf{x}(t)) = \mu \varphi(\mathbf{x}(t))$. Furthermore, the formula $\mu = \log(\lambda)/{dt}$ may emphasize the graphical feature of eigenvalue distribution due to its scale. In the following, our Koopman operator will be $\mathcal{K}^{dt}$ with a fixed time step $dt$, but we will use the formula $\mu = \log(\lambda)/{dt}$ when we describe our numeric results, and will estimate the differential equation above when we do Koopman analysis.

On the other hand, the set $\{\mathcal{K}^t\}_{t\geq0}$ forms a one-parameter semigroup, and when the dynamics is sufficiently smooth, it admits its infinitesimal generator $\mathcal{L}$. Note that $\mu$ is an eigenvalue of the infinitesimal generator $\mathcal{L}$. The operator $\mathcal{L}$ has been called the Lie operator, and its adjoint, $\mathcal{L}^\dagger$, is called the Liouville operator. We continue the discussion on these operators in Sec.~SIIIB and Sec.~SIIIC. 

\subsubsection{Dynamic Mode Decomposition}
The most basic form of DMD finds a linear matrix operator $A$ satisfying $\mathbf{x}(n+1) \approx A\mathbf{x}(n)$ for each $n$. This is a linear approximation of time-dependent modes of the dynamics and is unable to capture the nonlinear phenomena generally.
To tackle this issue, Williams et al. \cite{Williams15} introduced the Extended DMD, leveraging the usage of nonlinear observables. This Extended DMD is a finite-dimensional approximation of the Koopman operator and can be viewed as a Galerkin method.

For the given (finite) snapshot data $\{(\mathbf{x}_i,\mathbf{y}_i)\}_{1\leq i \leq M}$ satisfying $T(\mathbf{x}_i) = \mathbf{y}_i$, we will approximate the Koopman operator using Extended DMD. Let $\mathcal{D}=\{g_1,\ldots,g_L\}$ be the set of observables $g_i:\mathbb{M}\to \mathbb{C}$, called dictionary, and denote its vector form (with the index order) by $\vec{\mathcal{D}}=[g_1,\ldots,g_L]^T$ and the spanned space $V_L\coloneqq \mbox{span}(\mathcal{D})$. We note that the Koopman operator $\mathcal{K}^{dt}$ satisfies $\mathcal{K}^{dt}g_i(\mathbf{x}_j) = g_i(\mathbf{y}_j)$. To build matrix representations, we denote
\begin{align*}
    \Psi_X = \begin{bmatrix}
        \vec{\mathcal{D}}(\mathbf{x}_1) &\vec{\mathcal{D}}(\mathbf{x}_2)&\cdots & \vec{\mathcal{D}}(\mathbf{x}_M)
    \end{bmatrix}\in\mathbb{C}^{L\times M},\quad
    \Psi_Y = \begin{bmatrix}
        \vec{\mathcal{D}}(\mathbf{y}_1)&\vec{\mathcal{D}}(\mathbf{y}_2)&\cdots & \vec{\mathcal{D}}(\mathbf{y}_M)
    \end{bmatrix},
\end{align*}
where $\vec{\cal D}(\mathbf{x})= \begin{bmatrix}g_1(\mathbf{x}) &g_2(\mathbf{x})\cdots&g_L(\mathbf{x})\end{bmatrix}^T$. Now we seek a minimizer of the minimization problem 
\[ \min_{K\in \mathbb{C}^{L\times L}} \| \Psi_Y - K\Psi_X \|_2,\]
and denote one solution by $\mathbb{K} = \Psi_X^\dagger \Psi_Y$, which is the matrix representation of the Koopman operator in the Extended DMD context. We call the eigenvectors of $\mathbb{K}$ the Koopman modes (or DMD modes) in the main context, and each mode corresponds to a particular eigenvalue of $\mathbb{K}$ 

This extended DMD approximates the Koopman operator $\mathcal{K}^{dt}$ by $\mathcal{K}_{L,M}:V_L \to {\rm L}^2(\mathcal{M})$, defined as $\mathcal{K}_{L,M} \varphi \coloneqq (\mathbb{K}\mathbf{c}_\varphi)\cdot \vec{\mathcal{D}}$ for $\varphi = \mathbf{c}_\varphi \cdot \vec{\mathcal{D}} \in V_L$ and $\mathbf{c}_\varphi\in\mathbb{C}^L\setminus\{0\}$. According to the convergence theory \cite{Korda18}, appropriate data sampling justifies the convergence in the large data limit, $\lim_{M\to\infty} \|\mathcal{K}_{L, M}f - \mathcal{K}_Lf\|$ for $f\in V_L$ where $\mathcal{K}_L\coloneqq \mathcal{P}_L\mathcal{K}|_{V_L}$, $\mathcal{P}_L:{\rm L}^2\to V_N$ is the ${\rm L}^2$-projection, and the norm $\|\cdot\|$ is the any norm in $V_L = \mbox{span}(\mathcal{D})$. Again, appropriate assumption on dictionary gives us convergence $\|\mathcal{K}_L - \mathcal{K}^{dt} \|_{{\rm L}^2(\mathcal{M})}\to0$ as $L\to\infty$. Concerning the first limitation, we could take the essential supremum norm over $V_L$, and then the smoothness of observables we will use gives the supremum norm and pointwise convergence. Concerning the second limit in ${\rm L}^2$, we might not have pointwise convergence, but there exists a subsequence $\{ L_k\}$ such that $\mathcal{K}_{L_k}$ converges to $\mathcal{K}$ almost everywhere.

For an observable $\varphi\in\mathcal{F}$, we approximate its multi-step prediction along the dynamics as follows
\begin{align*}
\varphi(T^{(n)}(\mathbf{x}_1)) = [\mathcal{K}^{(n)}\varphi](\mathbf{x}_1) \approx [\mathcal{K}^{(n)}_L \varphi](\mathbf{x}_1) \approx [\mathcal{K}^{(n)}_{L,M}\varphi](\mathbf{x}_1).
\end{align*}

The residual DMD \cite{colbrook24-res} eliminates spurious eigenvalues whose residual exceeds the manually selected threshold $\varepsilon$. Precisely, we numerically compute the $\varepsilon$-pseudospectrum of the Koopman operator $\mathcal{K}$, $\sigma_\varepsilon(\mathcal{K})=\{\lambda\in\mathbb{C}:\|(\mathcal{A}-\lambda I)^{-1}\|\geq \varepsilon^{-1} \}$. The $\varepsilon$ pseudospectrum always contains the spectrum of the Koopman operator, $\sigma(\mathcal{K})$; $\sigma_\varepsilon(\mathcal{K}) \supset \sigma(\mathcal{K})$,  and, in fact, $\cap_{\varepsilon\geq0}\sigma_\varepsilon(\mathcal{K}) = \sigma(\mathcal{K})$ \cite{Trefethen05}.

Lastly, we introduce the exact DMD, which has become the standard DMD and leverages the singular value decomposition (SVD) to get computational accuracy and efficiency. Given data matrices $X,Y\in\mathbb{R}^{d\times M}$, DMD computes the eigendecompoisiton of the best-fit linear operator $A$ relating $Y\approx AX$ and set $A = YX^\dagger$ \cite{Tu14}. The eigenvectors of $A$ are called the DMD modes or dynamic modes. The exact DMD starts with the truncated SVD $X\approx U\Sigma V^*$, $U\in \mathbb{C}^{d\times r}$, $\Sigma\in\mathbb{C}^{r\times r}$, $V\in\mathbb{C}^{M\times r}$. Here, $*$ is the conjugate transpose, $r$ is less than or equal to the rank of the reduced SVD approximation to $X$, $U$ and $V$ have orthonormal columns. In this case, $A=YV\Sigma^{-1}U^*$, but in practice, it is efficient computationally to compute $\tilde{A}=U^*YV\Sigma^{-1}$. We compute the eigendecomposition $\tilde{A}W = W\Lambda$ where the columns of $W$ are eigenvectors and $\Lambda$ is a diagonal matrix of eigenvalues. Then, $\Lambda$ is the eigenvalues of $A$ and the columns of $\Phi = YV\Sigma^{-1}W$ are eigenvectors of $A$. Indeed, we can see
\begin{align*}
    A \Phi = YV\Sigma^{-1}\tilde{A} W = \Phi \Lambda.
\end{align*}

\subsection{Well-definedness for Kuramoto systems and basic property}
We closely look at the definition of the Koopman operator for the Kuramoto system as a composition operator. As before, we fix a positive constant $dt>0$.
\begin{lemma}[{\cite[Theorem 2.11]{Singh93} } ] \label{T1.1}
Let $(X,\mathcal{A},m)$ be a $\sigma$-finite measure space and let $T:X\to X$ be a measurable transformation. Then $T$ induces a bounded linear composition operator $U_T$ on ${\rm L}^2$ if and only if there exists a constant $C >0$ such that
\begin{align*}
m (T^{-1}(S)) \leq C \cdot m(S),\quad \forall~ S\in\mathcal{A}.
\end{align*}
\end{lemma}
\begin{proof}
The proof of this theorem may be found in \cite{Singh93}.
\end{proof}
We recall that our state space of the Kuramoto system is the measure space $(\mathbb{T}^N,\sigma(\mathbb{T}^N),m)$ where $\sigma(\mathbb{T}^N)$ is the Borel sigma algebra and $m$ is the (induced) Lebsegue measure on the tori $\mathbb{T}^N$. 
\begin{lemma}
The Koopman operator $\mathcal{K}^{dt}$ corresponding to the Kuramoto system is a bounded linear operator on ${\rm L}^2(\mathbb{T}^N)$.
\end{lemma}
\begin{proof}
Since the vector field of the Kuramoto system is real analytic, the solution operator $\Phi^{dt}:\mathbb{T}^N \to \mathbb{T}^N$ is a diffeomorphism
\begin{align}\label{A-1}
(\Phi^{dt})^{-1} \circ (\Phi^{dt}) = \mbox{id}_{\mathbb{T}^N}.
\end{align}
In the proof of Liouville's theorem, we have $m(\Phi^t(B)) = \int_B J\Phi^t(\Theta) dm(\Theta)$ for $B \in \sigma(\mathbb{T}^N)$ where $J\Phi^t(\Theta)$ is the Jacobian of $\Phi^t$ at $\Theta$. Furthermore, we have
\begin{align*}
\frac{d}{dt}m(\Phi^t(B)) 
&= \int_{\Phi^t(B)} (\nabla_\Theta \cdot \Phi^t) d\Theta\\
&=  -\int_{\Phi^t(B)} \sum_{j\neq i} A_{ji}\cos(\theta_j - \theta_i) ~d\Theta
\geq - J A_M N \cdot m(\Phi^t(B)),\quad t>0.
\end{align*}
Here, $A_M = \max_{i,j} |A_{ij}|$. Using Gr\"onwall's inequality, we have
\begin{align}\label{A-2}
m(\Phi^{dt}(B)) \geq e^{-J A_M N dt}m(\Phi^0(B)) = e^{-J A_M N dt}m(B).
\end{align}
We combine Eq. \eqref{A-1} and Eq. \eqref{A-2} to have
\begin{align*}
m\big( (\Phi^{dt})^{-1}(B)) e^{-J A_M N dt} \leq m\left( ((\Phi^{dt})^{-1} \circ \Phi^{dt} ) (B) \right) = m(B),\quad B\in\sigma(\mathbb{T}^N).
\end{align*}
We apply Lemma \ref{T1.1} with $C = e^{ J A_M N dt} $.
\end{proof}
On the other hand, one can generalize this theorem.
\begin{lemma}\label{L2.1}
Suppose that the state space $\mathcal{M}$ is compact, and the vector field $F$ and its gradient $\nabla F$ corresponding to the system in Eq. \eqref{SM-A-1} are sufficiently smooth and bounded, respectively. Then the corresponding Koopman operator $\mathcal{K}^{dt}$ with a number $dt>0$ is well-defined over ${\rm L}^2(\mathcal{M})$.
\end{lemma}

If the system preserves the measure (a measure-preserving system is such that $m(S)=m(T^{-1}(S))$ for all $S\in\mathcal{A}$), one might have the one-shot convergence algorithm to have the spectrum of the Koopman operator \cite{colbrook24-limit}. However, this is not the case for the Kuramoto system, as apparent by the irreversible decay to incoherence. On the other hand, if the Koopman operator is an isometry or unitary, it can be proven that the spectrum lies in the unit circle. We introduce an equivalent condition for a Koopman operator being an isometry below. Our numeric observations suggest that the Koopman operator of the Kuramoto system is not an isometry, and the spectrum is instead constrained to the unit disk (i.e., it can contain both oscillatory and decaying modes).
\begin{lemma}
Let $y = y(t;0,y_0)$ be the solution to Eq. \eqref{SM-A-1} passing through $y_0$ at $t=0$. Suppose that the transformation $\Phi^{dt}:\mathcal{M}\to \Phi^{dt}(\mathcal{M})$ induced by the flow is diffeomorphism and the Koopman operator $\mathcal{K}^{dt}:{\rm L}^p(\mathcal{M})\to {\rm L}^p(\mathcal{M})$ is well-defined. Then, the Koopman operator $\mathcal{K}^{dt}$ is an isometry if and only if
\begin{align}\label{A-3}
 \int_{0}^{dt} \nabla_y \cdot F( y(s;0,x)) ds =0,\quad \mbox{a.e.}~x\in\mathcal{M}.
\end{align}
\end{lemma}
\begin{proof}
The classical theory of the ordinary differential equation tells us that the derivative of the solution with respect to the initial data is given by the matrix, called the fundamental matrix, 
\[
\frac{\partial y}{\partial y_0}(t ; 0,y_0) = \Psi(t; 0,y_0) \eqqcolon \Psi(t)
\]
where $\Psi(t)$ is the resolvent of the variational equation
\begin{align*}
\dot\Psi(t) = \nabla_y F(y(t;0,y_0))\cdot \Psi(t),\quad \Psi(0) = I.
\end{align*}
On the other hand, Liouville's formula is
\begin{align*}
|\det \Psi(t;0,y_0) | = |\det \Psi(0) |\exp\left( \int_{0}^t \mbox{tr} \left( \nabla_y F( y(s;0,y_0)) \right) ds \right)>0.
\end{align*}
We note that $\frac{\partial \Phi^{dt}}{\partial x}(x) = \Psi(dt;0,x)$. For an observable $g\in {\rm L}^p(\mathcal{M})$, 
\begin{align}\label{A-4}
\begin{aligned}
\| \mathcal{K}^{dt} g\|_{{\rm L}^p}^p &= \int |g \circ \Phi^{dt}|^p dx = \int |g|^p \left|\det \frac{\partial \Phi^{-dt}}{\partial x}\right|^p dx = \int |g|^p \left|\det \frac{\partial\Phi^{dt}}{\partial x}\right|^{-p} dx\\
&= \int |g|^p \exp\left(-p \int_{0}^{dt} \mbox{tr} \left( \nabla_y F( y(s;0,x)) \right) ds \right) dx =  \int |g|^p \exp\left(-p \int_{0}^{dt} \nabla_y \cdot F( y(s;0,x)) ds \right) dx,
\end{aligned}
\end{align}
where we used the change of variable in the second equality. We note that Eq. \eqref{A-3} implies that $\|\mathcal{K}^{dt} g\|_{{\rm L}^p} = \| g\|_{{\rm L}^p}$.

If the Koopman operator $\mathcal{K}^{dt}$ is an isometry, then the integrand $\exp(\cdots)$ in Eq. \eqref{A-4} should be $1$ since we picked an arbitrary $g\in {\rm L}^p$.
\end{proof}
\noindent {\bf Example.} The Koopman operator corresponding to a Hamiltonian system is an isometry.

The following more general theorem, constraining the spectrum to the unit disk, is also easy to prove.
\begin{theorem}[{\cite[Thereom 5.1]{Bruce19} }] Define $\mathbb{T}^N_0 \coloneqq \{x\in\mathbb{T}^N:((\Phi^{dt})^{(k)}(x))_{k=0}^\infty)\mbox{ is bounded.}\}$. Let $\lambda$ be an eigenvalue of $\mathcal{K}^{dt}$ with corresponding eigenfunction $g$, and assume that $|\lambda|>1$ and $g$ is bounded on every bounded subset of $\mathbb{T}^N$. Then for all $x\in\mathbb{T}^N_0$, $g(x)=0$.
\end{theorem}
In our DMD simulations of the Kuramoto system, we consider the phase variables on tori and use dictionary functions that are continuous and bounded on the tori.
Since all trajectories on the tori are bounded, Theorem 5.1 applies. If a Koopman eigenfunction is represented as a (finite) linear combination of the (bounded) dictionary elements in our DMD setting, then it is necessarily bounded on the torus, and hence its associated eigenvalue satisfies $|\lambda|\leq1$. 
Although the eigenmodes obtained from extended DMD do not necessarily correspond to exact Koopman eigenfunctions \cite{Korda18}, the positive real parts observed in the logarithms of the (discrete-time) eigenvalues may reflect DMD spectral noise due to finite time of the trajectories and the non-invariance of the extended DMD library.

\subsubsection{Note on the continuous and point spectra of the Koopman operator}
We reconstruct the solution trajectory using eigenvalues of the approximated Koopman operator. However, an infinite-dimensional linear operator may have a nonempty point spectrum and a nonempty continuous spectrum simultaneously, and these affect the solution construction \cite{colbrook24-res}. Since we only have finite eigenvalues in each case, we display a graphical comparison between our reconstruction and the true trajectory to justify our reconstruction. If there is no continuous spectrum, then one might justify the solution reconstruction as a limit of a sequence. However, this does not appear to be the case for our Kuramoto systems.

The Koopman operators corresponding to each system we study have an infinite number of eigenvalues. If a linear operator is compact, then the cardinality of its eigenvalues is countable. However, it is well known that if the state space $(\mathcal{M},\mathcal{A},\mu)$ is a non-atomic measure space, then no composition operator on ${\rm L}^2(\mathcal{M})$ is compact (\cite{Singh93} Corollary 2.3.4). Thus, the Koopman operator corresponding to the Kuramoto system over ${\rm L}^2(\mathbb{T}^N)$ is not compact. 

\section{Details in the minimal example}
\subsection{Koopman analysis}
We first briefly describe analytic results concerning the Koopman spectrum for Eq.~1 of the main text. We note that general solutions to such one-dimensional dynamics can be easily found, which can be used to construct Koopman eigenfunctions \cite{Morrison24}. The specific functional forms in the case of Eq.~1 are not especially enlightening, so we instead emphasize the asymptotic behavior as the system approaches the $\theta=0$ attractor in the large time limit. In this limit, the equations can be approximated as 
\begin{equation*}
\dot{\theta}=-\theta^\zeta,
\end{equation*}
and the general solutions are $\theta(t)=\theta_0 e^{-t}$ for $\zeta=1$ and $\theta(t)=1/\sqrt{2t+{\theta_0}^{-2}}$ for $\zeta=3$, where $\theta(0)=\theta_0$. Then it follows that the measurement functions $g^{(1)}_\mu(x)\equiv x^{-\mu}$ and $g^{(3)}_\mu(x) \equiv e^{\mu/2x^2}$ will asymptotically evolve exponentially in the respective cases $\zeta=1$ and $\zeta=3$, and are thus candidates for Koopman eigenfunctions near $x=0$. However, for some values $\mu$, these functions are not sufficiently regular at $x=0$. In both cases, $\mathrm{Re}(\mu)<0$ is required to avoid blow up at $x=0$. In the exponentially decaying case, while $g^{(1)}_\mu(x)=x^{-\mu}$ does not blow up at $x=0$ if $\mathrm{Re}(\mu)<0$, it is not continuously differentiable at $x=0$ unless $\mathrm{Re}(\mu)<-1$. Thus, the regularity that one imposes on the eigenfunctions determines the size of the gap in the spectrum. In practice with DMD, the choice of dictionary functions and the numerical implementation will restrict which measurements can be well approximated, and thus impose such regularity requirements on the DMD spectrum.  

More formally, in Eq.~1, we can establish bounds 
($\mathcal{M}=(0,\pi/2)$)
\begin{align*}
&\zeta = 1;\quad -\theta(t) \leq \dot\theta(t) \leq -\frac{2}{\pi}\theta(t) \quad \Rightarrow\quad e^{-t} \theta(0) \leq \theta(t) \leq e^{-\frac{2}{\pi} t}\theta(0),\quad t\geq0,\\
&\zeta =3;\quad -\theta^\zeta(t) \leq \dot\theta(t) \leq -(\frac{2}{\pi})^\zeta \theta^\zeta(t) \quad \Rightarrow\quad \frac{(\zeta-1)^{-1}}{t + \frac{1}{(\zeta-1)\theta^{\zeta-1}(0)}}\leq \theta^{\zeta -1}(t)\leq \frac{(\zeta-1)^{-1}}{ (\frac{2}{\pi})^\zeta  t + \frac{1}{(\zeta-1)\theta^{\zeta-1}(0)}},\quad t\geq0.
\end{align*}
We study continuously differentiable eigenfunctions $\varphi=\varphi_r + \mathrm{i}\varphi_i$ of the Koopman operator satisfying that $|\varphi|$ is not constant, say $\mathcal{K}^t\varphi = e^{\mu t}\varphi$ for $\mu = \mu_r + \mathrm{i}\mu_i$.
For a positive constant $\ell$, we observe
\begin{align*}
\zeta =1;\quad 
\frac{d}{dt}\left( \frac{\varphi_r^2(\theta(t)) + \varphi_i^2(\theta(t))}{\theta(t)^{-\ell\mu_r}} \right) 
= 2\mu_r |\varphi(\theta(t))|^2\theta(t)^{\ell\mu_r}\left(1 + \frac{\ell}{2}\theta(t)^{-1}\dot{\theta}(t)\right),\quad t\geq0.
\end{align*}
Without loss of generality, we assume $\mu_r \leq 0$. When $\ell = \pi$ or $\ell = 1$, we have
\begin{align}\label{A-6}
\begin{aligned}
    &\zeta=1;\quad \frac{d}{dt}\left( \frac{\varphi_r(\theta(t))^2 + \varphi_i(\theta(t))^2}{\theta(t)^{-\pi\mu_r}} \right)\geq 0,\quad \frac{d}{dt}\left( \frac{\varphi_r(\theta(t))^2 + \varphi_i(\theta(t))^2}{\theta(t)^{-\mu_r}} \right)\leq 0\\ 
    &\Rightarrow\quad \frac{\varphi_r^2(\theta(0)) + \varphi_i^2(\theta(0))}{\theta(0)^{-\mu_r}} \theta(t)^{-\mu_r} \geq \varphi_r(\theta(t))^2 + \varphi_i(\theta(t))^2 \geq \frac{\varphi_r^2(\theta(0)) + \varphi_i^2(\theta(0))}{\theta(0)^{-\pi\mu_r}} \theta(t)^{-\pi\mu_r},\quad t\geq0.
\end{aligned}
\end{align}
Our claim is that $\varphi_r(\theta)^2 + \varphi_i(\theta)^2 \lesssim \theta^{-\mu_r}$ for $\theta\in\mathcal{M}$. On the other hand, the range of the solution trajectory $\theta(t;0,\pi/2)$ covers the interval $(0,\pi/2)$. For any $\theta\in(0,\pi/2)$, we take $t_1$ such that $\theta(t_1;0,\pi/2)=\theta$. We evaluate $t=t_1$ to Eq. \eqref{A-6} to have $\varphi_r^2(\theta) + \varphi_i^2(\theta) =\varphi_r^2(\theta(t_1)) + \varphi_i^2(\theta(t_1))\leq \frac{\varphi_r^2(\pi/2) + \varphi_i^2(\pi/2)}{(\pi/2)^{-\mu_r\pi/2}} \theta^{-\mu_r}$.
This holds for any $\theta\in\mathcal{M}$. In a similar way, we have $\theta^{-\pi\mu_r}\lesssim \varphi_r(\theta)^2 + \varphi_i(\theta)^2$. Thus, we have $\theta^{-\pi\mu_r}\lesssim |\varphi(\theta)|^2\lesssim \theta^{-\mu_r}$. Similarly, we have $e^{\mu_r/\theta^2}\lesssim |\varphi(\theta(t))|^2\lesssim e^{(\frac{2}{\pi})^3\mu_r/\theta^2} $ for $\zeta=3$.

\subsection{DMD convergence and reconstructions}
The top row in Fig.~\ref{figs1} shows the evolution of the DMD mode amplitudes for the five initial conditions presented in the main text. In both the $\zeta=1$ and $\zeta=3$ cases, the mode amplitudes decay exponentially for each trajectory, as expected. We use the Fourier series coefficients of the sawtooth function to reconstruct the trajectory from the reconstructed library terms as
\begin{equation}
\hat{x}^j(t)=\sum_k \frac{(-1)^k}{2k} \mathrm{Im}\left(\sum_i \tilde{\Phi}_i^k b_i^j(0) e^{\mu_i t}\right)
\end{equation}
The bottom row in Fig.~\ref{figs1} shows the evolution of the trajectories and the DMD reconstructions. We see that, aside from brief initial errors for some cases, the reconstructions very closely reproduce the trajectories.
\begin{figure}[hbt]
\includegraphics[width=0.75\columnwidth]{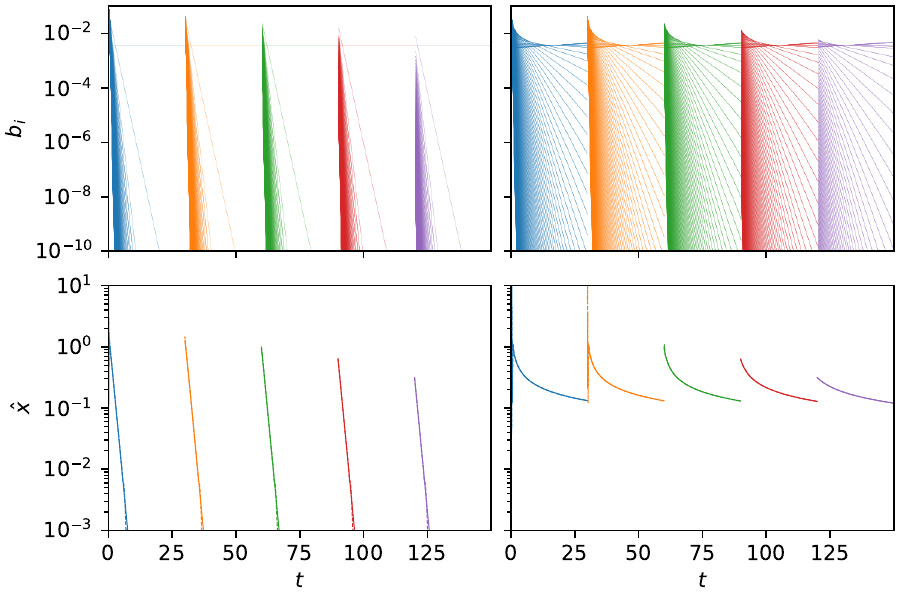}
\caption{DMD mode amplitude $b_i$  (top row) and  trajectory reconstruction $\hat{x}$ (bottom row) vs time for Eq.~1 of the main text with $\zeta=1$ (left column) and $\zeta=3$ (right column). The colors indicate different trajectories. In the bottom row, the reconstructions are marked by dashed lines and the actual trajectories are marked by continuous lines, but the difference between the two is almost indistinguishable except at the start of the first three trajectories.  \label{figs1}}
\end{figure}

Figure \ref{figs2} shows the distribution of the logarithm of the real part of the eigenvalues (left) and the amplitudes (right) for increasing dictionary size. If the scaling followed Ogielski and Stein's example in the main text, we would expect the data to lie on straight lines. While the scaling here differs somewhat, it is clear from the reconstructed trajectories that it still results in algebraic decay.
\begin{figure}[hbt]
\includegraphics[width=0.75\columnwidth]{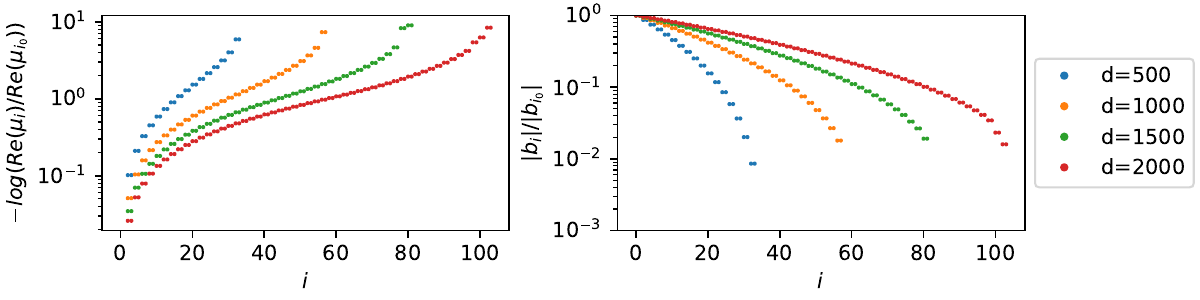}
\caption{(a) Logarithmic scaled eigenvalue vs mode index on a log-linear plot and (b) scaled mode amplitude vs mode index for the DMD spectrum in the $\zeta=3$ case of Eq.~1 of the main text.  \label{figs2}}
\end{figure}

\section{Details in the oscillator glass}
\subsection{Numerical integration}
Ottino and Strogatz \cite{strogatz} used a fixed timestep $dt=0.01$ with a 4th-order Runge-Kutta method on a CPU.  We implement instead an adaptive time-stepping Dormand-Prince with a 5/4 embedding Runge-Kutta method on the GPU (the same method used by the ode45 in Matlab and the Python scipy method solve\_ivp).  Our efficient simulations are capable of exploring systems as large as $N=10^6$,  which is three orders of magnitude larger than those previously considered. The trajectory is sampled at the $dt=0.01$ timesteps using a free fourth-order polynomial interpolation. We reproduce both the low $K=2$ clustered states and the apparent high-$K$ algebraic decay to incoherence for $N=5000$ oscillators in Fig.~\ref{fig1}(a).

\begin{figure}[hbt]
\includegraphics[width=0.7\columnwidth]{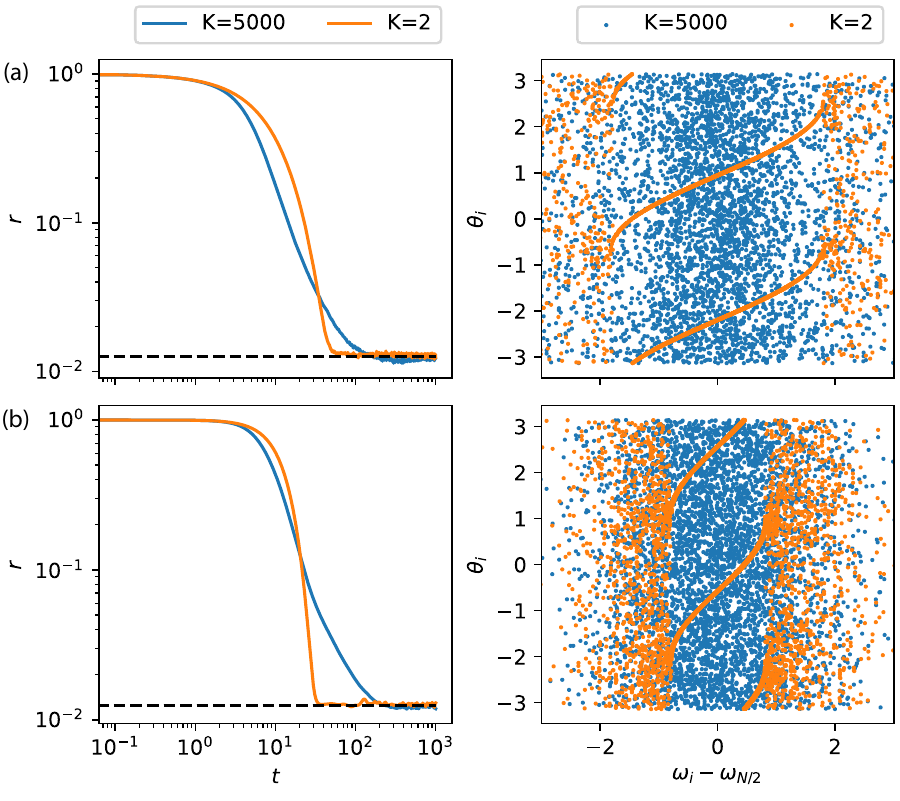}
\caption{Volcano clusters for small $K$ and slow decay to incoherence for large $K$ for Lorenztian distributed frequencies with $J=3$ (a) and normal distributed frequencies with $J=1.75$ (b), with $N=5000$.  The noise floor from Eq.~\eqref{incoherent} is shown by the dashed black line.  Order parameters are averaged over $750$ simulations. \label{figs3}}
\end{figure}

Our first observation is that the fixed time step used by Ottino and Strogatz is not sufficiently small to guarantee an accurate solution for the problem with Lorentzian distributed frequencies, even for the relatively small systems with $N=5000$ presented.  To control the error well, the timestep should be smaller than $O(T_{\mathrm{min}})$ where $T_{\mathrm{min}}=2\pi/\omega_{\mathrm{max}}$ is the natural period of the fastest oscillator since otherwise, the contributions to Eq.~2 will oscillate many times between time steps.  If the frequencies are distributed with a probability density $f(\omega)$, it can be shown that the distribution for the maximum frequency of $N$ samples is $Nf(\omega)F(\omega)^{N-1}$, where $F$ is the cumulative probability distribution, i.e. $F'(\omega)=f(\omega)$. So the expected maximum frequency is $\omega_{\mathrm{max}} = \int N\omega f(\omega)F(\omega)^{N-1} d\omega$. For the standard Lorenzian distribution, this gives $\omega_{\mathrm{max}} \sim 1.6475 N$ asymptotically for large $N$, leading to $T_{\mathrm{min}}=0.0076$ for $N=5000$.  Indeed, our adaptive time step uses $dt\approx0.0015$ for $K=N=5000$. This scaling poses a significant challenge to increasing $N$ with Lorentzian-distributed natural frequencies. On the other hand, as shown in Fig.~\ref{fig1}(b), we observe a similar volcano transition for normally distributed natural frequencies in the small $K$ case. We find that the critical coupling $J_c$ is smaller than the value of $2$ in the Lorentzian case, consistent with the recent theoretical prediction of $J_c=\sqrt{8/\pi}\approx1.60$ \cite{Pazo23}.   For the normal distribution, we numerically observe a much slower growth with $\omega_{\mathrm{max}}  < \log N$. Thus, we opt to use normally distributed natural frequencies throughout this work. In the normally distributed case, the adaptive time step does not become prohibitively small even for $N$ and $K$ up to $10^6$.

A second consideration for large $N$ and $K$ is the memory required to store the adjacency matrix. This requires at least $N\times K$ bits to store the $u_m(j)=\pm 1$ values, and, more simply, $N^2 \times 8$ bytes to store each floating point in the adjacency matrix.  Modern graphics cards typically have less than 20 GB of memory available, so this would constrain $N=K\lesssim 5\times 10^4$ if we store the entire adjacency matrix during the computation.  However, the numerical integration routine itself requires only $O(N)$ bytes of memory.  By careful use of the counter-based Philox64 random number generator, we can efficiently regenerate the numbers at an arbitrary position of a fixed random sequence without storing the entire sequence.  Thus, we can recalculate the adjacency elements in each CUDA kernel call only when required, and keep the memory requirements down to $O(N)$. Our limitation is then only the time required to integrate the system. In this way, it takes less than ten seconds to integrate $K=N=10^4$ oscillators and about two days to integrate $K=N=10^6$ oscillators with normally distributed natural frequencies. (Compare this with scipy integration on a CPU, which takes about two minutes for the $K=N=10^4$ case).

\subsubsection{Noise floor}
For large $N$ and $K$, there is no clustering and the oscillator phases are expected to decay towards an incoherent state which is uniformly random and independent, and the order parameter will decrease until finite size effects become relevant. We can analytically predict the finite-$N$ scaling, noting that the real and imaginary components of the complex order parameter will be distributed according to sums of $N$ independent random variables (sampled by taking the cosine and sine of a uniform distribution). For large $N$, by the central limit theorem, these sums will converge to two independent normal distributions, each with variance $1/2N$. The magnitude of the order parameter will then be distributed as a Rayleigh distribution with $\sigma=(2N)^{-1/2}$, and the mean magnitude of the order parameter is then 
\begin{equation}
r_{\mathrm{incoherent}}=\sqrt{\pi/4N}.  \label{incoherent}
\end{equation}
This noise floor limits the period over which exponential or algebraic decay can be observed. Larger systems thus enable longer observations of the decay towards incoherence.

\subsubsection{Systems with large $N$}
Next, we consider the transition between the exponential and algebraic decay as $K$ and $N$ increase. To top row of Fig.~\ref{figs4} shows ours results at fixed coupling constant $J$ for cases with $N\in(10^4, 2\times 10^4, 10^5,2\times10^5)$. We see that for intermediate $K/N$, there is a kink in the curves signaling a transition between exponential and algebraic decay around $t=O(1)$, with the time of onset of the transition increasing modestly with increasing $N$.  The rate of the algebraic decay also decreases with increasing $N$, because the difference between the $O(N^{-1/2})$ noise floor in Eq.~\eqref{incoherent} and the values of $r$ at the transition time becomes increasingly small. Thus, at fixed coupling constant, the purported algebraic decay appears to be confined to an increasingly short period as $N$ increases, while the rate of decay is also increasingly slow. 

The bottom row of Fig.~\ref{figs4} shows our results when we scale the couple constant according to $J\sim\sqrt{K}$. We observe that scaling the coupling constant in this way allows us to fix the time of kink signaling the transition between exponential and algebraic decay.  Thus, our findings and experiments suggest that to maintain longer periods of algebraic decay as $N$ increases in a thermodynamic limit, we should simultaneously increase the coupling constant as $J\sim \sqrt{K}$.  This observation may be important for future studies considering the thermodynamic limit of the system.
\begin{figure}[hbt]
\includegraphics[width=\columnwidth]{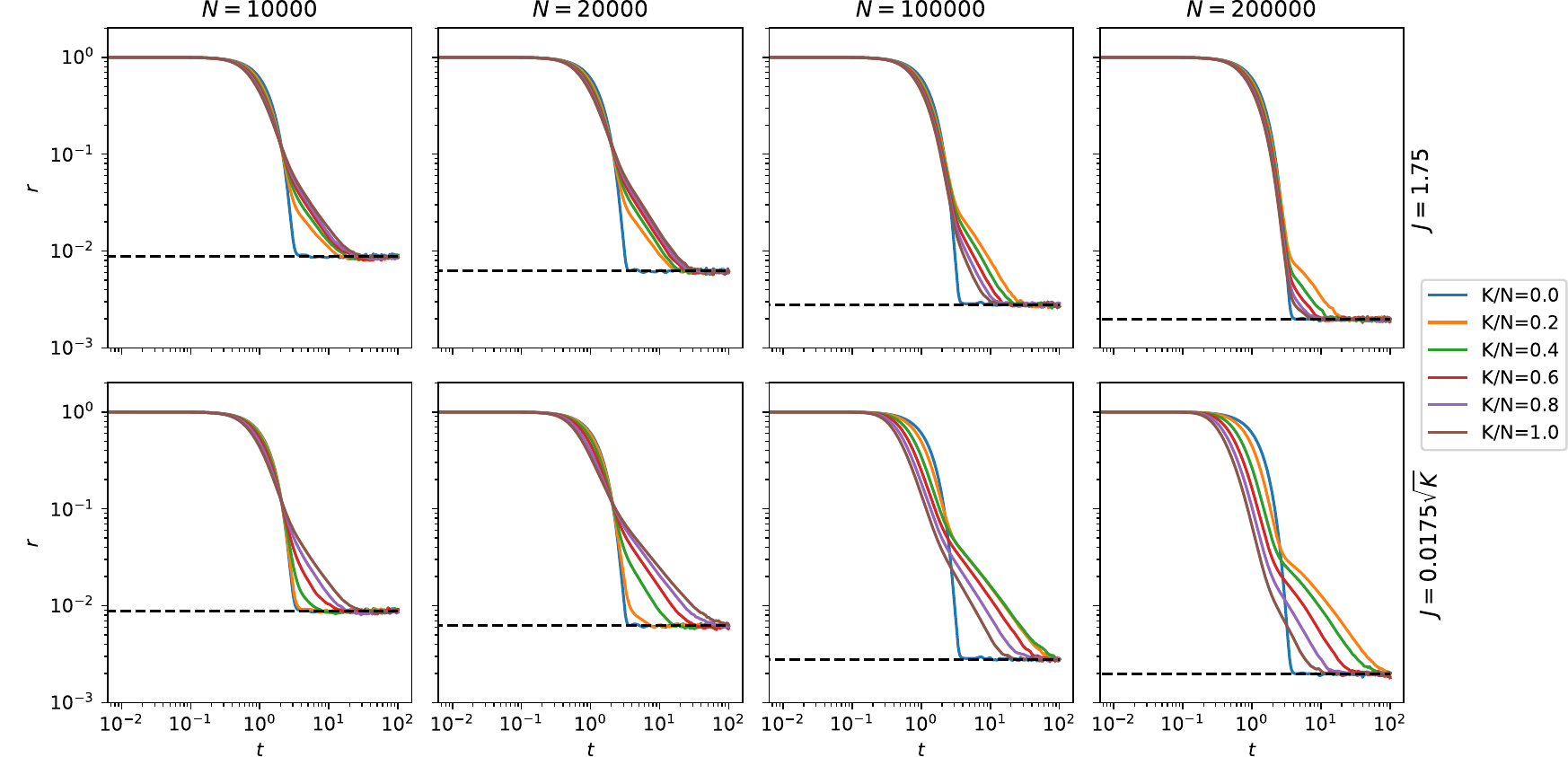}
\caption{Order parameter averaged over $500$ simulations vs time for various values of $N$, $K$, and $J$. \label{figs4}}
\end{figure} 

\subsection{Koopman analysis}
The flow induced by Eq.~2 (main text) on the torus ${\mathcal{T}}^N$ is generated by vector field 
\begin{equation}
\mathbf{v}=\sum_i \left(\omega_i + J\sum_j A_{ij} \sin\left(\theta_j-\theta_i)\right)\right)\mathbf{v}_i,
\end{equation}
where $\mathbf{v}_i=\frac{\partial}{\partial \theta_i}$ is the unit vector corresponding to the usual angle coordinate $\theta_i$ on ${\mathcal{T}}^N$.  The nonlinear dynamics can be equivalently described by the linear Liouville-von Neuman equation for a probability density $\rho : \mathcal{T}^N \to \mathbb{R}$,
\begin{equation}
\frac{\partial \rho}{\partial t} = {\cal L}^\dag_{\mathbf{v}}(\rho) \equiv -\sum_i \frac{\partial}{\partial \theta_i}\left[ \left(\omega_i + J\sum_j A_{ij} \sin\left(\theta_j-\theta_i)\right)\right)\rho\right], 
\end{equation}
or by the adjoint equation describing the infinitesimal flow of measurement quantities $g : \mathcal{T}^N \to \mathbb{R}$,
\begin{equation}
\frac{\partial g}{\partial t} = {\cal L}_{\mathbf{v}}(g) \equiv \sum_i \left(\omega_i + J\sum_j A_{ij} \sin\left(\theta_j-\theta_i)\right)\right)\frac{\partial g}{\partial \theta_i}.
\end{equation}
We can exponentiate these infinitesimal operators to find the Koopman operator ${\mathcal{K}}^\tau \equiv e^{\tau \mathcal{L}_{\mathbf{v}}}$ and the Perron-Frobenius operator $\mathcal{P}^\tau \equiv e^{\tau \mathcal{L}^{\dag}_{\mathbf{v}}},$ but the infinitesimal forms are simpler and convenient for our purpose here.  It is clear that any Koopman eigenfunction $\mathcal{K}^\tau[g]=\lambda g$ is also a Lie generator eigenfunction $\mathcal{L}_\mathbf{v}[g]=\mu g$ and vice versa, and the eigenvalues are related by $\lambda = e^{\mu \tau}$. 

We use the Fourier basis to express functions as infinite-dimensional vectors of Fourier coefficients. That is, we equate the ${\rm L}^2$ function $g(\theta)=\sum_{\bm{\mu}\in\mathbb{Z}^N} g_{\bm{\mu}} e^{ \imath \sum_i \mu_i\theta_i}$ with the $\ell^2$ coefficients $g_{\bm{\mu}}$, where $\bm{\mu}\in \mathbb{Z}^N$ spans over all lists of $N$ integers. Then, it is easy to show that the Lie generator $\mathcal{L}_\mathbf{v}$ acts on $g_{\bm{\mu}}$ as a matrix $\sum_{\bm{\nu}\in\mathbb{Z}^N} L_{\bm{\mu} \bm{\nu}}g_{\bm{\nu}}$ where 
\begin{equation}
L_{\bm{\mu} \bm{\nu}} = \imath \sum_i \omega_i \nu_i \delta_{\bm{\mu}}^{\bm{\nu}} + J \sum_{ij} A_{ij} \nu_i (\delta_{\bm{\mu}}^{\bm{\nu}+\bm{e}_j-\bm{e}_i}-\delta_{\bm{\mu}}^{\bm{\nu}-\bm{e}_j+\bm{e}_i})/2.
\end{equation}
Here, $\mathbf{e}_j$ is the vector with a unit entry in the $jth$ element and a zero in all others and $\delta_{\bm{\mu}}^{\bm{\nu}}=\prod_i \delta_{\mu_i}^{\nu_i}$ is the product of the Kronecker delta over all indices.    
Note that for $J=0$, the matrix $L_{\bm{\mu} \bm{\nu}}$ is diagonal, and indeed the Fourier basis elements $\hat{g}_{\bm{\mu}}(\theta) = e^{ \imath\sum_i \mu_i \theta_i}$ are Koopman eigenfunctions with (generator) eigenvalues $\mu =  \imath \sum_i \mu_i \omega_i$. 

We can alternatively consider temporal discretizations of the Koopman operator. 
For simplicity, we consider Euler's scheme: for a fixed $dt>0$, our discretized system is given by $\Theta_{n+1} = \Theta_n + h F(\Theta_n)$ where $F^i(\Theta_n) = \omega_n + J \sum_{j=1}^N J_{ji}\sin(\theta_n^j -\theta_n^i)$. The corresponding map $\Theta \mapsto \Theta + dtF(\Theta) \eqqcolon G(\Theta)$ is continuously differentiable and its derivative at $\Theta$ is given by
{\small
\begin{align*}
I_N + hJ
\begin{pmatrix}
-\sum_{\substack{ 1\leq j\leq N\\ j\neq1} } J_{j1}\cos(\theta^j - \theta^1) & J_{21}\cos(\theta^2-\theta^1) &\cdots & J_{N1} \cos(\theta_N -\theta_1) \\
J_{12}\cos(\theta^1-\theta^2) & -\sum_{\substack{ 1\leq j\leq N\\ j\neq2} } J_{j2}\cos(\theta^j - \theta^2)  &\cdots & J_{N2} \cos(\theta^N -\theta^2) \\
\vdots & \vdots &\ddots & \vdots \\
J_{1N}\cos(\theta^1-\theta^N) & J_{2N} \cos(\theta^N -\theta^2) &\cdots  &-\sum_{\substack{ 1\leq j\leq N\\ j\neq N} } J_{jN}\cos(\theta^j - \theta^N)
\end{pmatrix}.
\end{align*}}
We will see that $\inf_{\Theta\in\mathbb{T}^N} |\det({\rm D} G(\Theta))|>0$ for sufficiently small $dt>0$ such that the composition operator $\mathcal{K}_2^{dt}$ defined as $[\mathcal{K}_2^{dt}g](\Theta) = g( \Theta + dt F(\Theta))$ for $g\in {\rm L}^2(\mathbb{T}^N)$ is well-defined:
\begin{align*}
\int_{\mathbb{T}^N} |\mathcal{K}_2^{dt} g(\Theta)|^2 d\Theta= \int_{\mathbb{T}^N} |g( \Theta + dt F(\Theta) )|^2 d\Theta
= \int_{\mathbb{T}^N} |g(\tilde\Theta)|^2 |\det {\rm D}G(\tilde{\Theta})|^{-1} d\tilde\Theta < \infty.
\end{align*}
We denote the big matrix above by $H(\Theta)$ for each $\Theta$. Note that this matrix $H(\Theta)$ is uniformly componentwise bounded, $\sup_{\Theta\in\mathbb{T}^N}\max_{ij} |H_{ij}(\Theta)|<\infty$, so $\Theta\mapsto \det(I_N + hJH(\Theta))$ is well-defined. We observe that the assignment $\Theta\mapsto \det(I_N + dt JH(\Theta))=\det({\rm D}G(\Theta))$ is continuous over the compact set $\mathbb{T}^N$, so this assignment attains its minimum at $\Theta_0$: $\det(I_N + hJH(\Theta))\geq \det(I_N + dt JH(\Theta_0))$ for all $\Theta\in\mathbb{T}^N$. For sufficiently small $dt>0$, $\det(I_N + dt J(\Theta_0)) >c$ for some $c>0$ due to the continuity of the map $\det$ and the fact that $\det(I_N) = 1$.

Either using truncations in the Fourier basis or with the temporal discretization, we can consider calculating the Koopman spectrum analytically or numerically, but the high dimensionality poses significant challenges. Note that in the coupled case, the function $\hat{g}_{\bm{\mu}}$ remains a Koopman eigenfunction only for the case with constant multiindex $\mu_i = m$, since $m \sum_i \theta_i$ is a conserved quantity. Given our knowledge of the uncoupled eigenfunctions, a perturbative approach could be employed to investigate the other eigenfunctions for small coupling constants.  It would be interesting to compare these approximations with the data-driven DMD approximations in future studies.

\subsection{DMD convergence and reconstructions}
We first present in Fig.~\ref{figs5} the apparent convergence of the DMD spectrum as the dictionary size increases.
\begin{figure}[hbt]
\includegraphics[width=0.75\columnwidth]{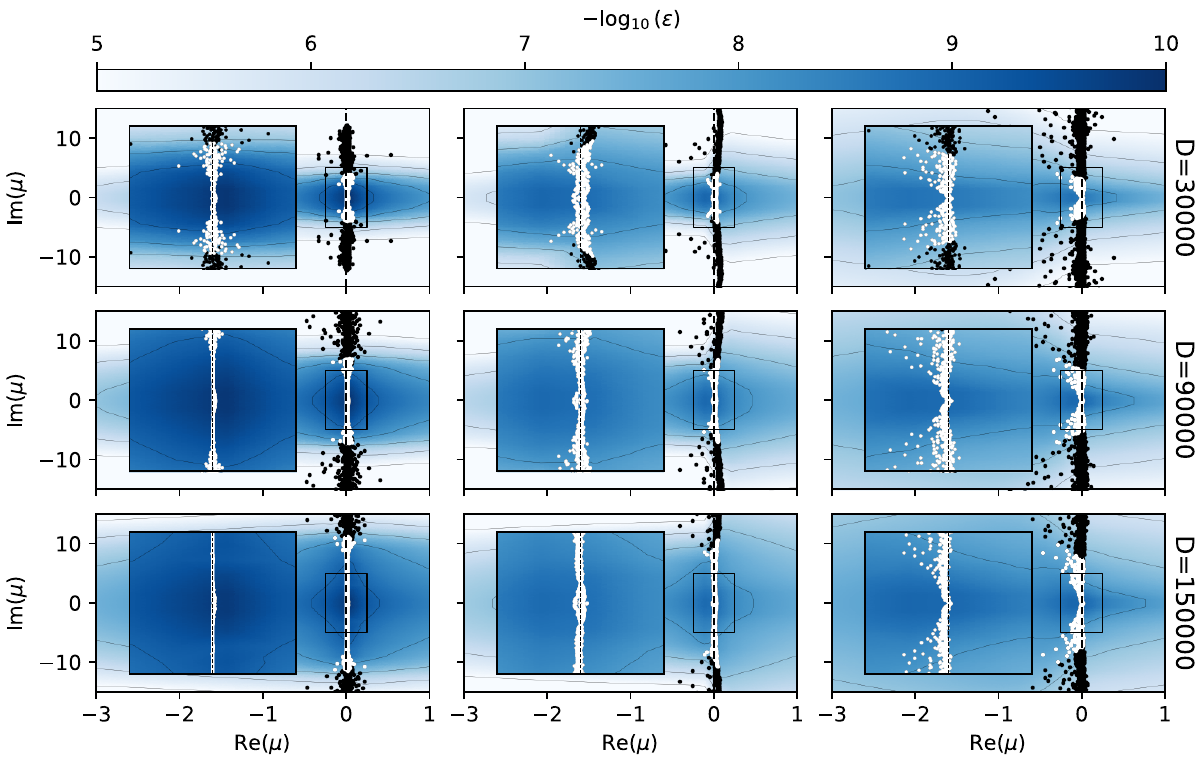}
\caption{DMD spectrum as in Fig.~4 of the main text for increasing dictionary size $D$. \label{figs5}}
\end{figure}
Next, we confirm that the DMD mode amplitudes evolve exponentially and that the reconstruction of the Kuramoto order parameter in Eq.~4 of the main text is very accurate, as shown in Fig.~\ref{figs6}.
\begin{figure}[hbt]
\includegraphics[width=0.75\columnwidth]{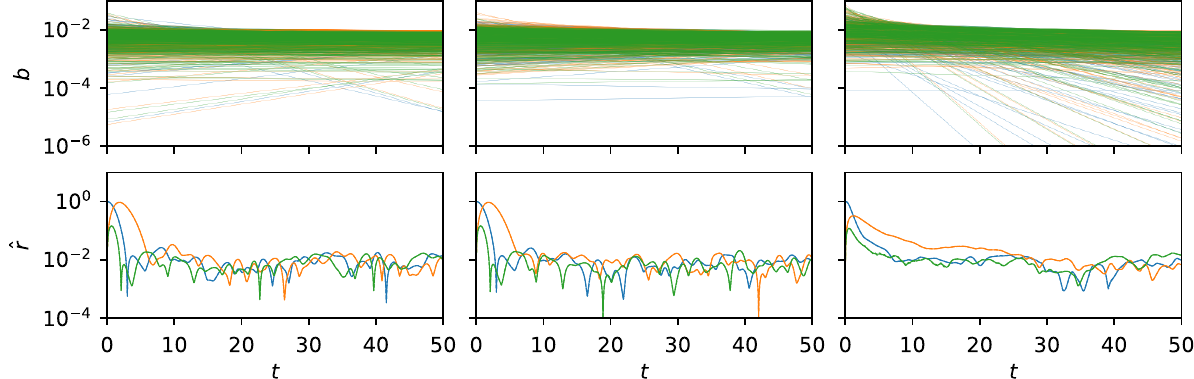}
\caption{DMD mode amplitudes $b_i$  as in Fig.~\ref{figs5} (top row) and trajectory (bottom row) actual order parameters (solid lines) and reconstructed order parameter (dashed lines) for the three parameter regimes (columns) and the three trajectories (line colors). The actual and reconstructed order parameters are nearly indistinguishable.\label{figs6}}
\end{figure}
Finally, in Fig.~\ref{figs7} we present the scaling between the magnitude of the real part of the eigenvalues and the reconstruction amplitudes for the DMD reconstructions of the glassy dynamics, which is qualitatively similar to the minimal example in Fig.~\ref{figs2}.
\begin{figure}[hbt]
\includegraphics[width=0.9\columnwidth]{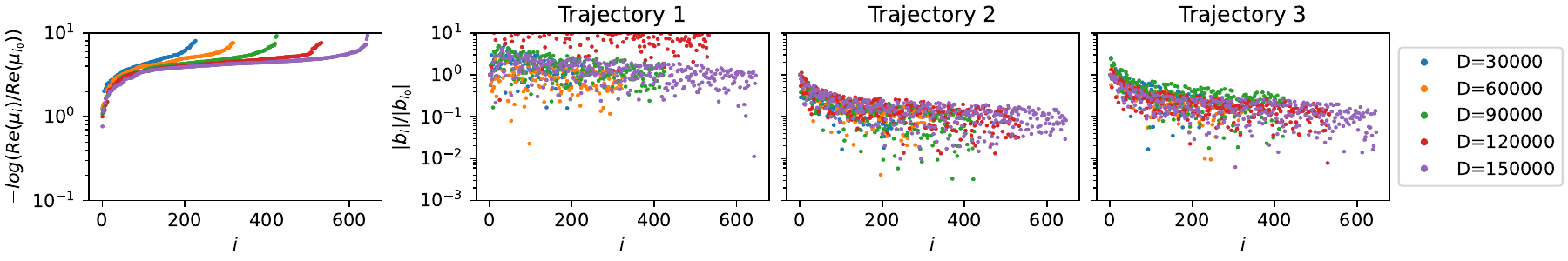}
\caption{Scaled DMD eigenvalue real part (left panel) and amplitudes for the three trajectories (right three panels) as in Fig.~\ref{figs2}. \label{figs7}}
\end{figure}

\subsection{Relation to the nonlinear Fokker-Planck equation}
Classical studies of the Kuramoto system often consider the number of oscillators with phase $\theta$ and natural frequency $\omega$ in a statistical ensemble via a density function $\hat{\rho}(\theta,\omega,t)$. In the large $N$ limit, the evolution of this density has been modeled with a nonlinear Fokker-Planck equation
\begin{equation}
\label{nlfp}
\frac{\partial}{\partial t}\hat{\rho}(\theta,\omega,t) = - \frac{\partial}{\partial \theta}\left[ \hat{\rho}(\theta,\omega,t) \left(\omega + K\int \hat{\rho}(\tilde{\theta},\tilde{\omega},t) \sin(\tilde{\theta} - \theta) ~d\tilde{\theta}d\tilde{\omega}\right)  \right].
\end{equation}
Here, we aim to demonstrate the relationship between Eq.~\eqref{nlfp} and our DMD analysis. To do so, we consider a statistical ensemble of $N$ oscillators with phases $\theta_i$ and natural frequencies $\omega_i$ with a density $\rho(\theta_i,\omega_i,t)$. It is easy to see that, just as in the case with fixed frequencies, the linear Liouville-von Neumann equation in Eq.~\eqref{lvme}, which we repeat here,
\begin{equation}
\frac{\partial}{\partial t}  \rho(\theta_i,\omega_i,t)= -\sum_i \frac{\partial}{\partial \theta_i}\left[ \left(\omega_i + J\sum_j A_{ij} \sin\left(\theta_j-\theta_i\right)\right)\rho(\theta_i,\omega_i,t)\right], \label{lvme}
\end{equation}
governs the evolution of this full density. (This follows because the natural frequencies in each microstate here do not vary in time. For example, we can simply extend the system state to include the natural frequencies as dynamical variables whose time derivatives vanish.) The full density $\rho(\theta_i,\omega_i,t)$ resides in a much higher dimension than the density $\hat{\rho}(\theta,\omega,t)$ in the Fokker-Planck equation and is capable of encoding complex correlations between the oscillator phases. A natural connection between the two follows from considering marginal distributions of the full density $\rho(\theta_i,\omega_i,t)$. 

We define the one-particle distribution function as
\begin{equation}
\rho^{(1)}(\theta, \omega, t)\equiv \sum_k \int \rho(\theta_i,\omega_i,t)\Big\rvert_{\substack{\theta_k= \theta \\ \omega_k=\omega}}~\overline{d\theta_kd\omega_k},
\end{equation}
where $\overline{d\theta_kd\omega_k}\equiv \prod_{i\neq k} d\theta_id\omega_i$ is the product of all the differentials except the $k$th ones. The time evolution of the one-particle distribution function can be directly evaluated from the definition and Eq.~\eqref{lvme},
\begin{align}
\frac{\partial}{\partial t}\rho^{(1)}(\theta, \omega, t) &= \sum_k \int \frac{\partial}{\partial t}\rho(\theta_i,\omega_i,t)\Bigg\rvert_{\substack{\theta_k= \theta \\ \omega_k=\omega}}~\overline{d\theta_kd\omega_k} =\dot{\rho}^{(1)}_{\omega}+\dot{\rho}^{(1)}_{\mathrm{interaction}} 
\end{align}
where
\begin{equation}
\dot{\rho}^{(1)}_{\omega} \equiv -\sum_k \int \sum_i \frac{\partial}{\partial \theta_i}\left[ \omega_i \rho(\theta_i,\omega_i,t)\right]\Big\rvert_{\substack{\theta_k= \theta \\ \omega_k=\omega}}~\overline{d\theta_kd\omega_k}, 
\end{equation}
and 
\begin{equation}
\dot{\rho}^{(1)}_{\mathrm{interaction}}  \equiv -\sum_k \int \sum_i \frac{\partial}{\partial \theta_i}\left[ \left(J\sum_j A_{ij} \sin\left(\theta_j-\theta_i)\right)\right)\rho(\theta_i,\omega_i,t)\right]\Big\rvert_{\substack{\theta_k= \theta \\ \omega_k=\omega}}~\overline{d\theta_kd\omega_k}. \label{rhoint}
\end{equation}

In both cases, terms with $i\neq k$ vanish by the divergence theorem, and we can consider only the $i=k$ terms in the sums. Thus, the natural frequency term can be simply expressed in terms of the one-particle distribution function
\begin{equation}
\dot{\rho}^{(1)}_{\omega} = \frac{\partial}{\partial \theta}\left(\omega\rho^{(1)}(\theta,\omega,t)\right).
\end{equation}

To evaluate the interaction term, we must introduce the two-particle distribution function
\begin{equation}
\rho^{(2)}(\theta, \omega, \tilde{\theta}, \tilde{\omega}, t)\equiv \sum_{k} \sum_{l\neq k} \int  \rho(\theta_i,\omega_i,t)\Bigg\rvert_{\substack{\theta_k= \theta \\ \omega_k=\omega \\\theta_l= \tilde{\theta} \\ \omega_l=\tilde{\omega}}}~\overline{d\theta_{kl}d\omega_{kl}}, \label{twopart}
\end{equation}
where $\overline{d\theta_{kl}d\omega_{kl}}\equiv \prod_{i\neq k,l} d\theta_id\omega_i$ now excludes both the $k$th and $l$th differentials.  
Furthermore, assuming that the oscillators are identically coupled and have identically distributed natural frequencies, we can make the ansatz that the full distribution function is invariant under permutations of the oscillator indices. (This invariance assumption naturally excludes distributions corresponding to states with spontaneously broken symmetries, as may be expected in the study of chimera states, for example.) In this case, we can ensure that each term in the sum in Eq.~\eqref{twopart} is identical, and thus
\begin{equation}
\int  \rho(\theta_i,\omega_i,t)\Bigg\rvert_{\substack{\theta_k= \theta \\ \omega_k=\omega \\\theta_l= \tilde{\theta} \\ \omega_l=\tilde{\omega}}}~\overline{d\theta_{kl}d\omega_{kl}} = \frac{1}{N(N-1)}\rho^{(2)}(\theta, \omega, \tilde{\theta}, \tilde{\omega}, t). 
\end{equation}
Plugging this into the interaction term in Eq.~\eqref{rhoint} and changing the dummy integration variables from $\theta_j$ and $\omega_j$ to $\tilde{\theta}$ and $\tilde\omega$, we find
\begin{align}
\dot{\rho}^{(1)}_{\mathrm{interaction}} &= -\sum_k \int  \frac{\partial}{\partial \theta_k}\left[ J\sum_j A_{kj} \sin\left(\theta_j-\theta_k)\right)\rho(\theta_i,\omega_i,t)\right]\Bigg\rvert_{\substack{\theta_k= \theta \\ \omega_k=\omega \\ \theta_j=\tilde\theta \\ \omega_k=\tilde\omega}}~\overline{d\theta_{kj}d\omega_{kj}}d\tilde{\theta}d\tilde{\omega} \\
&= \frac{\partial}{\partial \theta} \int \left( K \sin(\tilde{\theta}-\theta)\rho^{(2)}(\theta,\omega,\tilde{\theta},\tilde{\omega},t) \right) d\tilde{\theta}d\tilde{\omega},
\end{align}
where $K=\sum_{jk} \frac{J A_{jk}}{N(N-1)}$. Thus, the evolution of the one-particle distribution function depends on the two-particle distribution function
\begin{equation}
\frac{\partial}{\partial t}\rho^{(1)}(\theta,\omega,t) = -\frac{\partial}{\partial \theta}\left( \omega \rho^{(1)}(\theta,\omega,t) + K\int \sin(\tilde{\theta}-\theta)\rho^{(2)}(\theta,\omega,\tilde{\theta},\tilde{\omega},t) d\tilde{\theta}d\tilde{\omega}\right). \label{bbgky}
\end{equation}

Equation \eqref{bbgky} is a slight generalization of the BBGKY hierarchy in statistical mechanics, accounting for our added statistical description of the natural frequencies of the oscillators. To complete the derivation, we must introduce a closure approximation of independence in the two-particle distribution function,
\begin{equation}
\rho^{(2)}(\theta,\omega,\tilde{\theta},\tilde{\omega},t) \approx \rho^{(1)}(\theta,\omega,t) \rho^{(1)}(\tilde{\theta},\tilde{\omega},t).\label{ansatz}
\end{equation}
This independence approximation effectively ignores correlations between the phases of interacting oscillators and is thus a mean field approximation. We can expect that this approximation may be justified in the large $N$ limit because the correlations between individual pairs of oscillators will be drowned out by the interactions with all the other oscillators in the system. Introducing Eq.~\eqref{ansatz} into Eq.~\eqref{bbgky}, we arrive at the nonlinear Fokker-Planck equation in Eq.~\eqref{nlfp} when we relabel $\rho^{(1)}=\hat{\rho}$.

Concerning the validity of the independence closure approximation in Eq.~\eqref{ansatz}, we should note that such mean-field approximations are known to fail in low-dimensional systems when correlation length scales diverge and fluctuations become relevant around second-order phase transitions in, for example, the one-dimensional Ising model. Renormalization group approaches may be utilized in such systems. In the context of oscillator networks, the rank of the adjacency matrix plays a role similar to the spatial dimension in magnetic systems, as it determines the number of pairwise interactions that each oscillator experiences.  Thus, there may be cause for concern in applying this closure approximation to low-rank systems.

\end{document}